\documentclass[runningheads]{llncs}
\newif\ifanonymous
\anonymousfalse     
\usepackage[T1]{fontenc}
\usepackage{graphicx}
\usepackage{bbold}
\usepackage{stmaryrd}
\usepackage{amssymb}
\usepackage{amsmath}
\usepackage{csquotes}
\usepackage{xspace}
\usepackage{xcolor}
\usepackage{enumitem}
\usepackage[strings]{underscore}
\usepackage{graphicx}
\usepackage[caption=false]{subfig}
\usepackage{orcidlink}
\usepackage[misc,geometry]{ifsym}
\usepackage[nameinlink,noabbrev]{cleveref}

\usepackage[ruled,noend,linesnumbered]{algorithm2e}
\usepackage[noend]{algpseudocode}
\SetNlSty{textsf}{}{}
\SetAlgoNlRelativeSize{0}
\SetKw{Continue}{continue}
\definecolor{algcommentgray}{gray}{0.40} 

\SetCommentSty{mycommentfont}

\hypersetup{hidelinks}

\newcommand{\printlen}[1]{\typeout{LEN(#1)=\the#1}}
\newcommand{\mailto}{\textsuperscript{(\Letter)}}
\usepackage[]{todonotes}
\usepackage{xparse}

\definecolor{mygreen}{RGB}{23,140,105}
\definecolor{myorange}{RGB}{200,85,0}
\definecolor{mypurple}{RGB}{125,90,150}
\definecolor{myblue}{RGB}{25,105,165}

\usetikzlibrary{automata}
\usetikzlibrary{positioning}
\usetikzlibrary{arrows}
\usetikzlibrary{arrows.meta}
\usetikzlibrary{shapes.misc, positioning}
\usetikzlibrary{shapes}
\usetikzlibrary{decorations}
\usetikzlibrary{patterns}
\usetikzlibrary{matrix}

\tikzset{node distance=2.5cm,
         every state/.style={
         semithick,fill=gray!10},
         initial text={},
         double distance=2pt,
         every edge/.style={
            draw,
            ->,>=stealth',
            auto,
            semithick}}
\tikzset{every state/.append style={inner sep=0.0pt, minimum size=20pt}}
\tikzset{elliptic state/.style={draw,ellipse,semithick,fill=gray!10}}
\tikzset{ematrix/.style={double, double distance=4pt, >={Implies[length=5pt,width=8pt]}}}

\newcommand{\lang}{\ensuremath{\mathcal{L}}}

\newcommand{\syms}{\ensuremath{\Sigma}}
\newcommand{\syma}{\ensuremath{a}}
\newcommand{\symb}{\ensuremath{b}}

\newcommand{\words}{\ensuremath{\syms^\star}}
\newcommand{\eword}{\ensuremath{\varepsilon}}
\newcommand{\worda}{\ensuremath{w}}

\newcommand{\gens}{\ensuremath{\mathcal{F}}}
\newcommand{\gena}{\ensuremath{f}}

\NewDocumentCommand{\prefixes}{ o }{%
  \ensuremath{\mathit{prefixes}%
  \IfValueT{#1}{\!\left(#1\right)}}%
}
\NewDocumentCommand{\suffixes}{ o }{%
  \ensuremath{\mathsf{suffixes}%
  \IfValueT{#1}{\!\left(#1\right)}}%
}

\newcommand{\langdim}[1]{\ensuremath{| #1 |}}
\newcommand{\witdim}[1]{\ensuremath{| #1 |}}

\newcommand{\true}{\ensuremath{\mathit{true}}}
\newcommand{\false}{\ensuremath{\mathit{false}}}

\newcommand{\aut}{\ensuremath{\mathcal{A}}}
\newcommand{\hyp}{\ensuremath{\mathcal{H}}}
\newcommand{\states}{\ensuremath{Q}}
\newcommand{\state}{\ensuremath{q}}
\newcommand{\init}{\ensuremath{\iota}}
\newcommand{\trans}{\ensuremath{\delta}}
\newcommand{\fin}{\ensuremath{\lambda}}
\newcommand{\sem}[1]{\ensuremath{\llbracket#1\rrbracket}}

\newcommand{\semi}{\ensuremath{\mathbb{S}}}
\newcommand{\semiadd}{\ensuremath{\oplus}}

\newcommand{\semimul}{\ensuremath{\odot}}
\newcommand{\serimul}{\ensuremath{\odot}}
\newcommand{\semizero}{\ensuremath{\mathbb{0}}}
\newcommand{\semione}{\ensuremath{\mathbb{1}}}
\newcommand{\semia}{\ensuremath{s}}
\newcommand{\semib}{\ensuremath{s'}}
\newcommand{\seminat}{\ensuremath{\mathbb{N}}}

\newcommand{\semimod}[2]{\ensuremath{\mathcal{W}_{#1}\left(#2 \right)}}

\NewDocumentCommand{\serispan}{ o }{%
  \ensuremath{\mathit{span}%
  \IfValueT{#1}{\!\left(#1\right)}}%
}
\NewDocumentCommand{\fps}{ m m }{%
  #2\langle\!\langle #1 \rangle\!\rangle%
}

\NewDocumentCommand{\seris}{ o o }{%
  \IfNoValueTF{#1}{%
    \fps{\words}{\semi}%
  }{%
      \fps{#1}{#2}%
  }%
}
\NewDocumentCommand{\serisrat}{ o o }{%
  \seris[#1][#2]_{\mathrm{rec}}%
}

\NewDocumentCommand{\polybase}{ m m }{%
  #2\langle#1\rangle%
}

\NewDocumentCommand{\poly}{ o o }{%
  \IfNoValueTF{#1}{%
    \polybase{\words}{\semi}%
  }{%
      \polybase{#1}{#2}%
  }%
}

\NewDocumentCommand{\mem}{ o }{%
  \ensuremath{%
    \mathbf{out}%
    \IfValueT{#1}{\!\left(#1\right)}%
  }%
}
\NewDocumentCommand{\eq}{ o }{%
  \ensuremath{%
    \mathbf{eq}%
    \IfValueT{#1}{\!\left(#1\right)}%
  }%
}
\NewDocumentCommand{\solve}{ o }{%
  \ensuremath{%
    \mathbf{solve}%
    \IfValueT{#1}{\!\left(#1\right)}%
  }%
}

\NewDocumentCommand{\fencode}{ o }{%
  \ensuremath{%
    \mathsf{encode}%
    \IfValueT{#1}{\!\left(#1\right)}%
  }%
}

\NewDocumentCommand{\flearn}{ o }{%
	\ensuremath{%
		\mathsf{learn}%
		\IfValueT{#1}{\!\left(#1\right)}%
	}%
}

\NewDocumentCommand{\fgetwit}{ o }{%
	\ensuremath{%
		\mathsf{getWitness}%
		\IfValueT{#1}{\!\left(#1\right)}%
	}%
}

\NewDocumentCommand{\fgetaut}{ o }{%
	\ensuremath{%
		\mathsf{getAutomaton}%
		\IfValueT{#1}{\!\left(#1\right)}%
	}%
}

\NewDocumentCommand{\fdecode}{ o }{%
  \ensuremath{%
    \mathsf{decode}%
    \IfValueT{#1}{\!\left(#1\right)}%
  }%
}

\newcommand{\vcex}{\ensuremath{\mathit{res}}}

\newcommand{\vobs}{\ensuremath{O}}
\newcommand{\wit}{\ensuremath{\mathcal{W}}}
\newcommand{\unit}{\ensuremath{\mathbf{1}}}

\newcommand{\sat}[1]{\ensuremath{\mathsf{SAT}\left( #1 \right)}}
\newcommand{\model}{\ensuremath{\mathsf{model}}}

\newcommand{\teacher}{\ensuremath{\mathcal{T}}}

\newcommand{\eqsys}{\ensuremath{\Phi}}

\newcommand{\semibool}{\ensuremath{\mathbb{B}}}
\newcommand{\semitrop}{\ensuremath{\mathbb{T}}}
\newcommand{\semibtrop}[1]{\ensuremath{\mathbb{T}_{#1}}}
\newcommand{\semibot}{\ensuremath{\mathbb{Bot}}}
\newcommand{\badd}[1]{\ensuremath{+_{#1}}}

\newcommand{\theorynra}{\ensuremath{\textnormal{QF-NRA}}\xspace}
\newcommand{\theorylia}{\textnormal{\textnormal{QF-LIA}}\xspace}
\newcommand{\theorybv}{\ensuremath{\textnormal{QF-BV}}\xspace}

\usepackage{xspace}
\newcommand{\algopid}{\ensuremath{\mathsf{WL}^\star}\xspace}
\newcommand{\algonaive}{\ensuremath{\mathsf{Naive}}\xspace}
\newcommand{\algonew}{\ensuremath{\mathsf{SWAL}}\xspace}
\newcommand{\algolstar}{\ensuremath{\mathsf{L}^\star}\xspace}
\newcommand{\lstar}{\algolstar}
\newcommand{\algonlstar}{\ensuremath{\mathsf{NL}^\star}\xspace}

\begin{document}
\printlen{\textwidth}
\printlen{\linewidth}
\printlen{\columnwidth}
\title{\mbox{SMT-Based Active Learning of Weighted Automata}}
\ifanonymous
  \author{Anonymous}
  \authorrunning{Anonymous}
 \institute{}
\else
\author{Tiago Ferreira\inst{1}\mailto\orcidlink{0000-0002-6942-0228} \and
        Kevin Batz\inst{2}\orcidlink{0000-0001-8705-2564} \and
        Alexandra Silva\inst{2}\orcidlink{0000-0001-5014-9784}}
\authorrunning{Ferreira, Batz, and Silva}
\institute{University College London, London, UK\\
  \email{t.ferreira@ucl.ac.uk} \and
  Cornell University, Ithaca, NY, USA\\
  \email{\{ksb239,alexandra.silva\}@cornell.edu}}
\fi
\maketitle
\begin{abstract}
We present an SMT-based active learning algorithm for nondeterministic weighted automata (WFAs) as a practical and robust alternative to Hankel/\lstar-style methods. Our algorithm is parametric in a given semiring and, if it terminates, guaranteed to produce \emph{minimal} WFAs. We prove partial correctness and provide a sufficient termination condition, which in particular implies termination for all finite semirings. Our extensive experimental evaluation shows that our algorithm is capable of learning numerous minimal WFAs over both finite and infinite semirings, vastly outperforms a naive baseline, and is competitive with a state-of-the-art algorithm while producing significantly smaller automata and requiring less interaction with the teacher.
\end{abstract}

\section{Introduction}

Active automata learning, pioneered by Angluin's \lstar algorithm~\cite{angluin_learning_1987}, has become widely popular in formal verification~\cite{vaandrager_model_2017}. \lstar{}-style algorithms are routinely used to infer models of unknown components from queries, spanning applications in networking~\cite{fiterau-brostean_combining_2016,ferreira_prognosis_2021}, software verification~\cite{finkbeiner_learning_2015}, and real-world systems such as passports~\cite{aarts_inference_2010,marksteiner_automated_2024} or bank cards~\cite{aarts_formal_2013}. A key appeal of these algorithms is their tight integration with verification backends: counterexamples produced by model checkers naturally serve as equivalence-query answers~\cite{wu_black_1999}, while the learned automata act as compact, explainable abstractions of system behavior. Most of the work currently used in verification is based on Boolean automata: from the simpler deterministic finite-state model Angluin pioneered~\cite{angluin_learning_1987} to Mealy automata (input-output)~\cite{heerdt_efficient_2014} to automata with finite memory (with restricted access)~\cite{finkbeiner_scalable_2024}.

Weighted finite automata (WFAs) provide a uniform and expressive formalism for modeling and verifying quantitative system properties such as cost, probability, and resource consumption, making them a central abstraction in quantitative verification. They shift from Boolean to quantitative transitions, allowing these to carry weights from a semiring. This allows WFAs to model important quantities, and then serve as a basis to answer a range of important questions in quantitative verification and performance analysis: \emph{What is the minimal cost of reaching an error state?}, \emph{Is the probability of failure below a threshold?}, \emph{Does every execution respect a given resource bound?}. The question whether \lstar{}-style algorithms can be extended to WFA was a natural one to tackle and the last decade saw a range of theoretical results, mostly showing termination conditions for different classes of semirings: starting from fields~\cite{bergadano_learning_1996} to principal ideal domains~\cite{heerdt_learning_2020} or number rings~\cite{aristote_learning_2025}. Maybe not surprisingly, learning WFAs turned out to be substantially more challenging than learning deterministic finite automata. Classical Hankel- and table-based generalizations of \lstar must reason about algebraic structure in the weight domain (the semiring), and proving termination turned out to be very subtle over many natural semirings, often relying on indirect linear-algebraic arguments. In practice, these challenges limit both the applicability and the scalability of existing WFA learning algorithms.

This paper takes a different, algorithmically-driven approach to learning weighted automata. Rather than refining a Hankel matrix or observation table, we cast learning directly as a constraint-solving problem. Building on a coalgebraic view of WFAs, we derive a system of equations over the given semiring whose solutions correspond exactly to nondeterministic WFAs that agree with the teacher on all observed queries. Learning then amounts to incrementally strengthening this equation system in response to counterexamples and invoking an SMT solver to search for a solution. This perspective yields an algorithm that is parametric in the semiring, naturally supports nondeterminism, and --- when it terminates --- produces minimal weighted automata by construction.

A central advantage of this formulation is its compatibility with modern SMT technology. We show that the resulting equation systems fall into decidable arithmetic theories for several practically relevant semirings, including all finite semirings and infinite ones such as the tropical semiring. We further identify sufficient conditions guaranteeing termination, which in particular ensure termination for all finite semirings. An extensive experimental evaluation demonstrates that our SMT-based learner can infer numerous minimal WFAs over both finite and infinite semirings, outperforming a naive baseline and competing favorably with an existing algorithm while producing significantly smaller automata and requiring fewer interactions with the teacher. Taken together, our results position constraint-based learning as a practical and robust alternative to Hankel/\lstar-style methods for weighted automata in verification-oriented settings.

\paragraph*{\textbf{Contributions.}}
In a nutshell, our paper makes the following contributions:
\begin{enumerate}
\item We present a novel SMT-based active WFA learning algorithm (\Cref{sec:learning}), which, if it terminates, is guaranteed to provide \emph{minimal} WFAs. This guarantee of minimality on termination, which is semiring-agnostic, is a crucial advantage of our approach over \lstar{}-based methods where minimality guarantees only exist for very specific semirings (as we detail in \Cref{sec:related}).
	\item We prove partial correctness of our algorithm and provide sufficient conditions for termination, which in particular apply to all finite semirings. As a corollary, we obtain provable guarantees for learning minimal NFAs.	%
	\item At the heart of our algorithm is  the first constraint-based WFA generation algorithm enabling SMT-based automata synthesis (\Cref{sec:getaut}), and we identify semirings for which the  constraint \mbox{systems are decidable (\Cref{sec:decidable_semirings}).}
	\item We provide an extensive experimental evaluation in \Cref{sec:eval}, demonstrating that our algorithm is capable of learning numerous WFAs over both finite and infinite semirings. By comparing to a naive baseline as well as an existing WFA learning algorithm (which works for a restricted class of semirings), we show how our algorithm advances the state-of-the-art in active learning for weighted automata over a variety of semirings.
\end{enumerate}
We discuss related work in \Cref{sec:related} and conclude in \Cref{sec:conclusion}.

\section{Semirings and Weighted Automata}
\label{sec:prelim}
In this section, we briefly recap preliminary material on weighted automata and weighted languages, which are parametric on a semiring.
\begin{definition}[Semiring] A \emph{semiring} is a tuple
	\( \langle \semi, \semiadd, \semimul, \semizero, \semione\rangle \),
	where
$\semi$ is a {\em carrier} set,  $\langle \semi, \semiadd, \semizero\rangle$ is a commutative monoid,
 $\langle \semi,\semimul,\semione\rangle$ is a monoid,
  multiplication distributes over addition, i.e., for all $\semia_1,\semia_2,\semia_3\in\semi$,
		\[
			\semia_1 \semimul (\semia_2 \semiadd \semia_3) ~{}={}~ (\semia_1 \semimul \semia_2) \semiadd (\semia_1\semimul\semia_3)
			\qquad\text{and}\qquad
			(\semia_2 \semiadd \semia_3) \semimul \semia_1 ~{}={}~ (\semia_2 \semimul \semia_1) \semiadd (\semia_3\semimul\semia_1),
		\]
and $\semizero$ is annihilating, i.e., for all $\semia \in \semi$, we have $\semia\semimul \semizero = \semizero = \semizero \semimul \semia$.
\end{definition}
In this paper, we primarily focus on the following semirings~\cite[Sec. 2.1]{droste_handbook_2009}:
\begin{enumerate}[leftmargin=*]
	\item The \emph{Boolean semiring} $\semibool \triangleq \langle \{0,1\}, \vee, \wedge, 0, 1\rangle$.
	\item The \emph{Tropical semiring} $\semitrop \triangleq \langle \mathbb{N} \cup \{\infty\}, \min, +, \infty, 0\rangle$.
	\item The family of \emph{bounded Tropical semirings}, defined for a bound $b \in \seminat$ as
	 \(
	\semibtrop{b} \triangleq \langle \{0,\ldots,b \} \cup \{\infty\}, \min, \badd{b}, \infty, 0\rangle~,
	\),
	where $\semia \badd{b} \semib \triangleq \infty$ whenever $\semia + \semib > b$.
	\item The \emph{Bottleneck semiring} $\semibot \triangleq \langle \mathbb{N} \cup \{-\infty, \infty\}, \max, \min, -\infty, \infty  \rangle$.
\end{enumerate}
Below, to define weighted automata, we will use $
\semi$-valued matrices which we denote by $\semimod{\semi}{R\times C}$ where $R$ and $C$ are the sets indexing the rows and columns. When $R$ (resp. $C$) is a singleton set (denoted by $\unit$) the matrix is simply an $\semi$-valued row (resp. column) vector. Formally, for a set $X$, $$\semimod{\semi}{X} \triangleq \{ v \mid v\colon X \to \semi\}.$$ We will use $\times$ for standard matrix multiplication of these objects. For the rest of the paper, we fix a semiring $\langle \semi, \semiadd, \semimul, \semizero, \semione\rangle$, and a finite alphabet $\syms$.

\begin{definition}[WFA]
A \emph{weighted finite automaton~\cite{droste_handbook_2009} over $\semi$ ($\semi$-WFA)} is a tuple $\aut = \langle \states, \init, \trans, \fin \rangle$ where $\states$ is a finite nonempty set of states, $\init \in \semimod{\semi}{\unit \times \states}$ is the initial weight vector, $\trans$ is a $\syms$-indexed family of transition matrices $\trans^{\syma} \in \semimod{\semi}{\states\times\states}$, and $\lambda \in \semimod{\semi}{\states\times \unit}$ is the final weight vector.
\end{definition}
When $\semi$ is clear from context, we simply write WFA instead of $\semi$-WFA. WFAs recognize {\em weighted languages}, functions $\lang \colon \words \to \semi$, where $\words$ denotes the set of finite sequences (words) over $\syms$. Note how for the Boolean semiring the definition of WFA coincides with that of nondeterministic finite automata (NFAs), and {\em weighted languages} over $\semibool$ are languages in the usual sense (sets of words $2^{\words}$).

In the sequel, we denote by $\varepsilon \in \words$ the empty word and define the concatenation of words $\worda$ and $\worda'$ by $\worda \cdot \worda'$, often abbreviated by their juxtaposition $\worda\worda'$. We say that $\worda' \in \words$ is a suffix of $\worda$, if there exists $\worda''$ such that $\worda''\worda' = \worda$. We denote the set of suffixes of $\worda$ by $\suffixes[\worda]$. The set of all weighted languages is denoted by $\seris$.

\begin{definition}[Semantics]
The \emph{semantics} of a WFA $\aut$ is the weighted language $\sem{\aut} \in \seris$ given by:
\[
	\sem{\aut}(\worda) \triangleq
	\begin{cases}
		\init \times \fin &\text{if $\worda = \eword$}\\
	 \init \times \trans^{\syma_1} \times \ldots \times  \trans^{\syma_n} \times \fin & \text{if $\worda = \syma_1\ldots\syma_n$}
	 \end{cases}
\]
We call $\sem{\aut}(\worda)$ the \emph{weight $\aut$ assigns to $\worda$}.
\end{definition}
 A language $\lang$ is \emph{recognizable}, denoted $\lang \in \serisrat$, if there is a WFA $\aut$ recognizing $\lang$, i.e., $\sem{\aut} = \lang$.

A WFA $\aut$ is \emph{minimal}, if there does not exist a WFA $\aut'$ with fewer states and $\sem{\aut} = \sem{\aut'}$.
We denote by $\langdim{\lang}$ the \emph{dimension of $\lang \in \serisrat$}, which is the number of states of a minimal automaton recognizing $\lang$. Note that minimal WFAs are not unique in the same sense as deterministic finite automata --- this is because there might be different transition structures over the same number of states recognizing the same language. However, the number of states of the minimal automaton is a uniquely determined quantity.
\begin{example}
\label{ex:aut1}
	Consider the following $2$-state WFA $\aut \triangleq \langle \{q_1,q_2\}, \init, \trans,\fin \rangle$ over the tropical semiring $\semitrop$ and the alphabet $\{\syma, \symb\}$:

\begin{minipage}[t]{0.45\textwidth}
\vspace{-10pt}
        \centering
	\[
	\begin{aligned}
\init    &\triangleq \begin{pmatrix} \textbf{\textcolor{myorange}{4}} & \textbf{\textcolor{myorange}{5}} \end{pmatrix}
	  &\quad
	  \trans^a &\triangleq \begin{pmatrix} \textbf{\textcolor{myblue}{6} }& \infty \\ \infty &  \textbf{\textcolor{myblue}{4}} \end{pmatrix}
	  \\
\trans^b &\triangleq \begin{pmatrix} \textbf{\textcolor{mygreen}{6}} & \textbf{\textcolor{mygreen}{8}} \\ \infty & \textbf{\textcolor{mygreen}{12}} \end{pmatrix}
	  &\quad
\fin     &\triangleq \begin{pmatrix} \infty \\ \textbf{\textcolor{mypurple}{3}} \end{pmatrix}
	\end{aligned}
	\]
    \end{minipage}%
    \hfill
\begin{minipage}[t]{0.45\textwidth}
\vspace{0pt}
        \centering
	\begin{tikzpicture}[bend angle=12]

	  \node[state] (q1) {$q_1$};
	  \node[state, right of=q1, node distance=60pt] (q2) {$q_2$};

\node[draw=none, above of=q1, node distance=28pt] (i1) {$\textbf{\textcolor{myorange}{4}}$};
	  \draw[->] (i1) -- (q1);
\node[draw=none, above of=q2, node distance=28pt] (i2) {$\textbf{\textcolor{myorange}{5}}$};
	  \draw[->] (i2) -- (q2);

\node[draw=none, below of=q2, node distance=28pt] (o2) {$\textbf{\textcolor{mypurple}{3}}$};
	  \draw[->] (q2) -- (o2);

	  \draw (q1) edge[loop left]
	    node {$\begin{tabular}{c}
$a/ \textbf{\textcolor{myblue}{6}}$\\
$b/\textbf{\textcolor{mygreen}{6}}$
	    \end{tabular}$}
	    (q1);

	  \draw (q2) edge[loop right]
	    node {$\begin{tabular}{c}
$a/ \textbf{\textcolor{myblue}{4}}$\\
$b/\textbf{\textcolor{mygreen}{12}}$
	    \end{tabular}$}
	    (q2);

	  \draw (q1) edge[above, bend left]
	    node {$\begin{tabular}{c}
$b/\textbf{\textcolor{mygreen}{8}}$
	    \end{tabular}$}
	    (q2);

	\end{tikzpicture}
	\vspace{5pt}
    \end{minipage}

\noindent On the right, we depict the automaton as a state diagram and use colored weights to highlight the correspondence.  $\semizero$-valued transitions are omitted (and for the Tropical semiring $\semitrop$ we have $\semizero=\infty$).
The weight $\sem{\aut}(\syma\symb)$ that $\aut$ assigns to $\syma\symb$ is
\[
\begin{aligned}
\sem{\aut}(\syma\symb)
&{}=
\init \times \trans^{\syma} \times \trans^{\symb} \times \fin =
\begin{pmatrix} \textbf{\textcolor{myorange}{4}} & \textbf{\textcolor{myorange}{5}} \end{pmatrix}
\begin{pmatrix}
  \textbf{ \textcolor{myblue}{6}}  & \infty \\
  \infty &  \textbf{\textcolor{myblue}{4}}
\end{pmatrix}
\begin{pmatrix}
\textbf{\textcolor{mygreen}{6}} & \textbf{\textcolor{mygreen}{8}} \\
\infty & \textbf{\textcolor{mygreen}{12}}
\end{pmatrix}
\begin{pmatrix}
\infty \\ \textbf{\textcolor{mypurple}{3}}
\end{pmatrix}
= 21~.
\end{aligned}
\]
\end{example}

\section{Learning Algorithm}
\label{sec:learning}
In this section, we give a high-level overview of our learning algorithm (shown in \Cref{alg:learn}), which will take as input a {\em teacher} holding a weighted language $\mathcal L$, and return, if it terminates, a minimal WFA recognizing $\mathcal L$.

\subsection{The Teacher}
\label{sec:learning:basic}
We start with some basic notions. We fix an alphabet $\syms$ and a recognizable language $\lang \in \serisrat$ throughout this section. We employ the following standard notion of \emph{teachers} --- a component our algorithm is \mbox{assumed to have access to.}

\begin{definition}[Teacher]
	\label{def:teacher}
	A \emph{teacher} (for $\lang$) is a pair $\teacher \triangleq \langle \mem, \eq\rangle$, where:
	\begin{enumerate}
		\item\label{def:teacher1} $\mem \colon \words \to \semi$ is a computable function satisfying $\mem(\worda) = \lang(\worda)$;
		\item\label{def:teacher2} given a WFA $\aut$, $\eq(\aut) \in \{\true\} + \words$ is computable and satisfies
		\begin{enumerate}[topsep=5pt]
		\item
			$\eq(\aut) = \true ~\text{iff}~ \aut~\text{recognizes}~\lang$
			\quad\text{and}\quad
			\item $\eq(\aut)=\worda ~\text{implies}~\sem{\aut}(\worda) \neq \lang(\worda)~.$
		\end{enumerate}
	\end{enumerate}
\end{definition}
Invoking the teacher to determine $\mem(\worda)$ for some $\worda \in \words$ is called an \emph{output query}. Invoking $\eq(\aut)$ is called an \emph{equivalence query}.

We call (possibly infinite) nonempty sets $\vobs\subseteq \words$ of words \emph{observation sets}, which give rise to the straightforward notion of $\vobs$-correct WFAs --- WFAs that correctly recognize $\lang$ when restricting to words in $\vobs$:
\begin{definition}[$\vobs$-Correctness]
	$\aut$ is an $\vobs$-correct WFA (for $\lang$), if
	\[
	\forall \worda \in \vobs\colon \quad \sem{\aut}(\worda) = \lang(\worda)~.
	\]
\end{definition}

\subsection{The Main Loop of the Algorithm}
\label{sec:learning:algo}
\begin{algorithm}[t]
	\caption{$\flearn[\teacher]$}
	\label[algorithm]{alg:learn}
	\KwData{A teacher $\teacher = \langle \mem, \eq\rangle$ for $\lang$.}
	\KwResult{Minimal WFA $\hyp$ recognizing $\lang$.}
	$\vobs \gets \{\eword\}$; $n \gets 1$; $\hyp \gets \false$\label{learning_init}\;
	\While{$\true$}{
		\tcp{Loop Invariant: $\vobs$ is a finite observation set}
		\tcp{and $\hyp\neq \false$ implies $\hyp$ is an $n$-state $\vobs$-correct WFA}
		$\hyp \gets \fgetaut[\mem, \vobs, n]$\label{learning_getaut}\;
		\If{$\hyp=\false$}{
			\tcp{No $n$-state $\vobs$-correct WFA exists}
			$n \gets n + 1$; \textbf{continue}} \label{learning_continue}
		\tcp{$\hyp$ is an $n$-state $\vobs$-correct WFA}
		$\vcex \gets \eq(\hyp)$\;
		\If{$\vcex = \true$}{
			\tcp{$\hyp$ recognizes $\lang$}
			\Return $\hyp$; \label{learning_term}
		}
		\tcp{$\vcex$ is a counterexample to $\hyp$ recognizing $\lang$}
		$\vobs \gets \vobs \cup \{\vcex\}$\;
	}
\end{algorithm}
Fix a teacher $\teacher \triangleq \langle \mem, \eq\rangle$ for $\lang$ throughout this section. Our learning algorithm $\flearn[\teacher]$, depicted in \Cref{alg:learn}, takes $\teacher$ as input and, if it terminates, is guaranteed to produce a \emph{minimal} WFA recognizing $\lang$.

The main loop of the algorithm constructs new hypothesis automata using a subroutine $\fgetaut$ that satisfies the following high-level specification.
\begin{definition}[Specification of $\fgetaut$]
		\label{def:getaut}
	 The procedure
	\[
	\fgetaut[\mem, \vobs, n],
	\]
	given access to output queries via $\mem$, when it terminates, maps each finite observation set $\vobs$ and $n\geq 1$ to either $\false$ or a WFA $\hyp$. The implementation of the procedure is \emph{sound} if:
	\begin{enumerate}
		\item\label{def:getaut1} $\fgetaut[\mem, \vobs, n]  = \false$ iff there is no $n$-state $\vobs$-correct WFA,
		\item\label{def:getaut2} $\fgetaut[\mem, \vobs, n]  = \hyp$ implies $\hyp$ is an $n$-state $\vobs$-correct WFA.
	\end{enumerate}
	Moreover, if it always terminates then we say it is complete.
\end{definition}
Note that, in this section, $\fgetaut$ is treated as a black-box. Later, we will evaluate two distinct ways for implementing $\fgetaut$ (cf.\ \Cref{sec:getaut}).
Let us now explain how  \Cref{alg:learn} works. We first initialize the finite observation set $\vobs \gets \{\eword\}$ and $n \gets 1$~in~l.\ref{learning_init}. The subsequent loop, on every iteration, invokes $\fgetaut$ to determine whether an $n$-state $\vobs$-correct WFA $\hyp$ exists (l.\ \ref{learning_getaut}). If not, we increase $n$ (l. \ref{learning_continue}). Otherwise, we ask the teacher whether $\hyp$ is a sought-after WFA recognizing $\lang$ (l.\ \ref{learning_term}). If so, we return $\hyp$ (l. \ref{learning_term}). Otherwise, $\vcex$ will be a counterexample --- a word in $\words$ such that $\sem{\hyp}(\worda) \neq \lang(\worda)$ --- which we add to the observation set $\vobs$. This ensures that, indeed, we construct \enquote{increasingly correct} WFAs (w.r.t.\ the growing observation set $\vobs$).

Based on the specification of $\fgetaut$, we consider partial correctness and sufficient termination criteria for the algorithm.
\begin{theorem}[Partial Correctness of \Cref{alg:learn}]
	\label{thm:flearn_correct}
	Let $\fgetaut$ be sound. If $\flearn[\teacher]$ terminates, then
	\[
		\flearn[\teacher] = \hyp
		\qquad\text{implies}\qquad
		\text{$\hyp$ is a minimal WFA recognizing $\lang$}~.
	\]
\end{theorem}
\begin{proof}
	If $\flearn[\teacher] = \hyp$, then \Cref{alg:learn} terminates in l.\ \ref{learning_term}.  Hence, $\hyp$ recognizes $\lang$ by \Cref{def:teacher}.\ref{def:teacher2}. It remains to show that $\hyp$ is minimal. We proceed by contradiction. For that, assume $\hyp$ is \emph{not} minimal, i.e., the final value of $n$ in \Cref{alg:learn} exceeds $\langdim{\lang}$. Consequently, some call to $\fgetaut[\mem, \vobs, \langdim{\lang}]$ returned $\false$. By \Cref{def:getaut}, there is thus no  $\langdim{\lang}$-state $\vobs$-correct WFA, which is a contradiction since such a WFA exists for all $\vobs \subseteq \words$. \qed
\end{proof}

Next, we provide a sufficient termination condition for our learning algorithm.
\begin{theorem}[Sufficient Termination Condition for \Cref{alg:learn}]
	\label{thm:flearn_terminate}
	Let $\fgetaut$ be sound and complete. If, for all $n\geq 1$, there does not exist an infinite sequence $(\vobs_1, \hyp_1), (\vobs_2, \hyp_2), \dots$ of $n$-state $\vobs_i$-correct WFAs $\hyp_i$ such that
	\[
		\vobs_1 ~{}\subsetneq{}~ \vobs_2 ~{}\subsetneq{}~ \ldots~,
	\]
		then \flearn[\teacher] terminates.
\end{theorem}
\begin{proof}
	 First observe that, for every $n$ encountered during execution, \Cref{alg:learn} eventually reaches l.\ \ref{learning_continue} or l.\ \ref{learning_term} since, otherwise, the algorithm would produce an infinite chain of \enquote{increasingly $\vobs$-correct} $n$-state WFAs.

	 For all $n<\langdim{\lang}$, it reaches l.\ \ref{learning_continue} since terminating in l.\ \ref{learning_term} would contradict \Cref{thm:flearn_correct}: For $n < \langdim{\lang}$, there is no $n$-state WFA recognizing $\lang$. Consequently, \Cref{alg:learn} eventually reaches $n = \langdim{\lang}$. In that case, it eventually terminates in l.\
	\ref{learning_term} because terminating in l.\ \ref{learning_continue} contradicts \Cref{def:getaut}: For all $\vobs$, there exists an $\langdim{\lang}$-state $\vobs$-correct WFA.\qed
\end{proof}

\begin{corollary}[Termination for Finite Semirings]
	Take $\fgetaut$ to be sound and complete. If $\semi$ is finite, then $\flearn[\teacher]$ terminates.
\end{corollary}
\begin{proof}
	The condition from \Cref{thm:flearn_terminate} is trivially satisfied since, for every $n\geq 1$, there exist only finitely many $n$-state WFAs over a finite semiring.\qed
\end{proof}

\begin{remark}[Beyond finite semirings] Since a sound and complete $\fgetaut$ can be implemented in a brute-force manner for all finite semirings, it follows that \Cref{alg:learn} is correct and terminates for all finite semirings.  The reader might wonder if one can prove soundness and completeness for more general semirings, since finiteness is a stringent restriction. One might have to impose further restrictions on the sequence $\hyp_1, \hyp_2, \ldots$ of hypotheses to guarantee termination \emph{beyond} finite semirings. Indeed, if $\semi$ is finite, then for fixed $n$ there are only finitely many $n$-state WFA up to isomorphism (finitely many choices of initial/final weights and transition weights). Hence no infinite sequence $\hyp_1$, $\hyp_2$, \ldots exists at all. If $\semi$ is not finite, there are infinitely many distinct $n$-state WFA, so a pure cardinality argument disappears, and completeness is much more nuanced. This is in contrast with \lstar algorithms where correctness based on Hankel matrices provides a relation between the state spaces of subsequent hypotheses $\hyp_1$, $\hyp_2$, \ldots appearing in the algorithm, and for which algebraic properties such as the semiring being a field or Noetherian can then be used to argue termination. As we will describe in \Cref{sec:related}, the study of termination for these \lstar-like algorithms is subtle and has resulted in a variety of recent results.
\end{remark}

As we will see later in \Cref{sec:eval}, even though we do not guarantee termination, our implementation is capable of learning minimal WFAs over several infinite semirings, even for examples in which the existence of complete learning algorithms is open (e.g., the Tropical or the Bottleneck semiring).

\section{Computing $n$-State $\vobs$-Correct WFAs}
\label{sec:getaut}
In this section, we present two procedures --- a naive baseline as well as a more sophisticated technique inspired by a coalgebraic perspective on WFAs --- to implement $\fgetaut$ as specified in \Cref{def:getaut}. Both procedures reduce the problem of determining the existence of $n$-state $\vobs$-correct WFAs to a \emph{satisfiability problem of equality constraints} over the given semiring $\semi$.

We will show that sound and complete implementations of $\fgetaut$ exist if the underlying semiring is \emph{decidable} --- these are semirings for which the satisfiability of {\em finite equation systems} is decidable. Put more formally,
finite systems $\eqsys$ of equations over $\semi$ adhere to the grammar
\begin{align*}
	\eqsys ~{}\longrightarrow{}~
	 t = t ~|~ \eqsys \wedge \eqsys 
	\qquad\quad
	t  ~{}\longrightarrow{}~
	\semia \in \semi
	~|~
	 x
	 ~|~
	 t \semiadd t~|~ t \semimul t~,
\end{align*}
where $x$ is taken from some set of variables. This yields the following notion:
\begin{definition}[Decidable Semiring]
	\label{def:decidable}
	We say that a semiring $\langle \semi, \semiadd, \semimul, \semizero, \semione\rangle$ is \emph{decidable}, if there is a computable function
	$
	\sat{\eqsys}
	$
	mapping each finite system $\eqsys$ of equations to either (i) $\false$ iff $\eqsys$ has no solution, or (ii) a function $\model$ from the variables in $\eqsys$ to values in $\semi$, representing a satisfying \mbox{assignment for $\eqsys$.}
\end{definition}
Every finite semiring is trivially decidable by brute-forcing all possible solutions to $\eqsys$. Interestingly, various \emph{infinite} semiring are decidable as well by reducing to decidable arithmetical theories. All the example semirings we consider are decidable.
We will provide two distinct ways of obtaining the following result:
\begin{theorem}
	\label{thm:getaut_for_decidable}
	For every decidable semiring, there is a sound and complete implementation of $\fgetaut$.
\end{theorem}
For the semirings considered in this paper, we employ SMT solving over decidable theories to find solutions to said systems of equations, or for proving their absence. The respective encodings are presented in \Cref{sec:decidable_semirings}.

\subsection{A Naive Baseline Implementation}
\label{sec:getaut:naive}
Let $n \geq 1$ and let $\vobs$ be a nonempty finite observation set.
Our naive baseline implementation of $\fgetaut[\mem,n,\vobs]$ is derived directly from the semantics of WFAs presented in \Cref{sec:prelim}: We treat the initial weight vector $\init$, the transition matrices $\trans^{\syma}$, and the final weight vector $\fin$ --- all of which having the fixed dimension $n$ --- as unknowns and solve the finite equation system
\[
	\eqsys(n,\vobs) ~{}\triangleq{}~ \bigwedge_{\worda\in\vobs} \mem[\worda] =
	 \begin{cases}
		\init \times \fin &\text{if $\worda = \eword$}\\
		\init \times \trans^{\syma_1} \times \ldots \times  \trans^{\syma_k} \times \fin & \text{if $\worda = \syma_1\ldots\syma_k$}~.
	\end{cases}
\]
It is immediate that $\sat{\eqsys(n,\vobs)} = \false$ iff there is no $n$-state $\vobs$-correct WFA, and that, otherwise, every model of $\eqsys$ gives rise to an $n$-state $\vobs$-correct WFA. Hence, if the semiring $\semi$ is decidable, \Cref{thm:getaut_for_decidable} follows.
\begin{example}
\label{ex:enc-naive}
For the $\semitrop$-WFA of \Cref{ex:aut1}, and $\vobs$ = \{\syma\symb\}:
\[
\eqsys(2,\vobs) ~{}\triangleq{}~ 21 =
\begin{pmatrix} \init_1 & \init_2 \end{pmatrix}
\times
\begin{pmatrix}
\trans^\syma_{1,1} & \trans^\syma_{1,2} \\
\trans^\syma_{2,1} & \trans^\syma_{2,2}
\end{pmatrix}
\times
\begin{pmatrix}
\trans^\symb_{1,1} & \trans^\symb_{1,2} \\
\trans^\symb_{2,1} & \trans^\symb_{2,2}
\end{pmatrix}
\times
\begin{pmatrix}
\fin_1 \\ \fin_2
\end{pmatrix}
\]
We can see that the formula encodes directly the correctness of the automaton's weights, and therefore its language, w.r.t. the word $\syma\symb$.
\end{example}

\subsection{An Optimized Implementation via State Languages}
\label{sec:getaut:non_naive}
\noindent
The central idea for our main approach is to exploit a coalgebraic perspective on the semantics of WFAs~\cite{bonsangue_sound_2013} which enables a compositional view of the language of an automaton as a linear combination of languages of its individual states. Formally, for $\aut = \langle \states, \init, \trans, \fin \rangle$, we define $\sem{\aut,\state}$, the \emph{state language} of $\state \in \states$ as
\[
		\sem{\aut,\state}(\eword) \triangleq \fin(\state)
		~\text{and}~
		\sem{\aut,\state}(\syma\worda) \triangleq \trans^{a}_{\state,\state_1} \semimul \sem{\aut,\state_1}(\worda) \semiadd \ldots \semiadd \trans^{a}_{\state,\state_n} \semimul \sem{\aut,\state_n}(\worda)~.
\]
This yields the following relationship with the language $\sem{\aut}$ of $\aut$:
\begin{theorem}
	\label{thm:autsem_alt}
	Let $\aut = \langle \states, \init, \trans, \fin \rangle$ be a WFA with $\states = \{\state_1,\ldots,\state_n\}$.  Then:
	\[
		\forall \worda \in \words\colon \quad \sem{\aut}(\worda) = \init(\state_1) \semimul \sem{\aut,\state_1}(\worda) \semiadd \ldots \semiadd \init(\state_n) \semimul \sem{\aut,\state_n}(\worda)~.
	\]
\end{theorem}
State languages are naturally captured recursively, i.e., if $\worda = \eword$, then $\sem{\aut,\state}(\eword)$ is $\state$'s output weight $\fin(\state)$ since there are no symbols to be consumed. Otherwise, i.e., if $\worda = \syma\worda'$, then
$\sem{\aut,\state}(\syma\worda')$ is obtained from \enquote{consuming} $\syma$, and proceeding by running $\aut$ on $\worda'$ from the successors of $\state$. Given the state languages  $\sem{\aut,\state}(\worda)$ at every state $\state\in\states$, we obtain the weight $\sem{\aut}(\worda)$ by taking the initial weights given by $\init$ into account.
\begin{example}
For automaton $\aut$ from \Cref{ex:aut1} we compute below $\sem{\aut}(\syma\symb)$ through its state languages.
\begin{align*}
\sem{\aut, q_1}(\eword) \triangleq \fin(q_1) = \infty \qquad \sem{\aut, q_2}(\eword) \triangleq \fin(q_2) = 3
\end{align*}
\vspace{-1.7\baselineskip}
\begin{alignat*}{3}
\sem{\aut, q_1}(\symb) &{}={} \trans^\symb_{1,1} \semimul \sem{\aut, q_1}(\eword) \semiadd \trans^\symb_{1,2} \semimul \sem{\aut, q_2}(\eword) &{}={} \;6\; \semimul \infty \semiadd\,\; 8 \semimul 3 &= 11 \\
\sem{\aut, q_2}(\symb) &{}={} \trans^\symb_{2,1} \semimul \sem{\aut, q_1}(\eword) \semiadd \trans^\symb_{2,2} \semimul \sem{\aut, q_2}(\eword) &{}={} \infty \semimul \infty \semiadd 12 \semimul 3 &= 15 \\
\sem{\aut, q_1}(\syma\symb) &{}={} \trans^\syma_{1,1} \semimul \sem{\aut, q_1}(\symb) \semiadd \trans^\syma_{1,2} \semimul \sem{\aut, q_2}(\symb) &{}={} 6 \semimul 11 \semiadd \infty \semimul 15 &= 17 \\
\sem{\aut, q_2}(\syma\symb) &{}={} \trans^\syma_{2,1} \semimul \sem{\aut, q_1}(\symb) \semiadd \trans^\syma_{2,2} \semimul \sem{\aut, q_2}(\symb) &{}={} \infty \semimul 11 \semiadd 4 \semimul 15 &{}={} 19 \\
\end{alignat*}
\vspace{-2.4\baselineskip}
\begin{align*}
\sem{\aut}(\syma\symb) = \init(q_1) \semimul \sem{\aut, q_1}(\syma\symb) \semiadd \init(q_2) \semimul \sem{\aut, q_2}(\syma\symb) = 4 \semimul 17 \semiadd 5 \semimul 19 = 21
\end{align*}
\end{example}

In our approach to learning we will explore \Cref{thm:autsem_alt} to derive an encoding of a more state-local characterization of the existence of $n$-state $\vobs$-correct WFAs. Our experiments in \Cref{sec:eval} reveal this characterization to be highly beneficial when compared to the naive baseline from the previous section. Towards formalizing this characterization, we introduce:
\begin{definition}[\vobs-Witness]
	\label{def:wit}
	Let $\vobs$ be an observation set. An \mbox{\emph{$\vobs$-witness} (of $\lang$)} is a triple
	\[
		\wit ~{}\triangleq{}~ \langle \gens, \init, \trans\rangle,
	\]
	where:
	\begin{enumerate}[topsep=0pt]
		\item $\gens = (\gena_i)^n_{i=1}$ is an indexed family of $\vobs$-\emph{generators} with type $\suffixes[\vobs]\to\semi$,

	\item\label{def:wit:stab} $\trans$ is a $\syms$-indexed family of matrices $\trans^{\syma} \colon \semimod{\semi}{\gens\times\gens}$ satisfying
			\[
			\underbrace{\forall \gena_i \in \gens\colon \forall\syma\worda \in \suffixes[\vobs] \colon\quad \gena_i(\syma\worda) = \trans^\syma_{i,1} \semimul \gena_1(\worda) \semiadd \ldots \semiadd \trans^\syma_{i, n} \semimul \gena_n(\worda)}_{\text{$\trans$ \emph{stabilizes $\gens$ for $\vobs$}}}~,
			\]
		\item\label{def:wit:span} $\init \colon \semimod{\semi}{1\times \gens}$ is a weight vector satisfying
		\[
		\underbrace{\forall \worda \in \vobs \colon\quad
		\lang(\worda) = \init_1 \semimul \gena_1(\worda) \semiadd \ldots \semiadd \init_n \semimul \gena_n(\worda)}_{\text{$\init$ \emph{spans} $\lang$ in $\gens$ for $\vobs$}}
		\]
	\end{enumerate}
\end{definition}
Notice that we do not require $\vobs$ to be finite.
We call the number $n$ of generators in $\gens$ the \emph{dimension of $\wit$} and denote it by $\witdim{\wit}$. Intuitively, they correspond to the state languages in a WFA, with the imposed constraints on $\gens$, $\trans$, and $\init$ making sure that, when restricting to words in $\vobs$, these state languages are indeed captured by a WFA-structure.
Hence, by \Cref{thm:autsem_alt}, there is a one-to-one correspondence between witnesses and WFAs: Every $n$-dimensional $\vobs$-witness yields an $n$-state $\vobs$-correct automaton, and vice versa. It is in this sense that $\wit$ \emph{witnesses} the existence of the sought-after WFAs. Put formally:
\begin{theorem}[WFA-Witness Correspondence]
	\label{thm:wfa_wit_corr}
	Let  $n\geq 1$ and let $\vobs$ be an observation set.  We have:
	\begin{enumerate}
		\item\label{thm:wfa_wit_corr1} If  $\wit = \langle \gens, \init, \trans\rangle $ is an $n$-dimensional $\vobs$-witness, then
		\(
			\aut_\wit {}={} \langle \gens, \init,\trans,\fin\rangle\),
			{where} \(
			\fin(\gena_i) = \gena_i(\epsilon)
		\)
		is an $n$-state $\vobs$-correct WFA.
		\item\label{thm:wfa_wit_corr2} If $\aut =  \langle \{\state_1,\ldots,\state_n\}, \init, \trans, \fin \rangle$ is an $\vobs$-correct WFA, then the triple
			$\wit_\aut = \langle (\gena_i)^n_{i=1}, \init, \trans \rangle$
		is an $n$-dimensional $\vobs$-witness, where $\gena_i \colon \suffixes[\vobs] \to \semi$ are uniquely determined by
		\(\gena_i(\varepsilon) = \fin(\state_i) \) and \(
			\gena_i(\syma \worda) \triangleq \trans^\syma_{i,1} \serimul \gena_1(\worda) \semiadd \cdots \semiadd \trans^\syma_{i,n} \serimul \gena_{n}(\worda)
		\), for all $\syma\worda \in \suffixes[\vobs]$.
	\end{enumerate}
\end{theorem}
Notice that observation sets being nonempty implies $\eword \in\suffixes[\vobs]$, which is why the output weight vector $\fin$ can be readily read off the weights $\gena_i(\eword)$.

With \Cref{thm:wfa_wit_corr}, it is immediate that the existence of $n$-state $\vobs$-correct automata can be characterized via the existence of $n$-dimensional $\vobs$-witnesses:
\begin{corollary}[Existence of $\vobs$-Correct WFA]
	\label{cor:wfa_exists}
	Let  $n\geq 1$ and let $\vobs$ be an observation set. The following statements are equivalent:
	\begin{enumerate}
		\item There exists an $n$-dimensional $\vobs$-witness  $\wit = (\gens, \init, \trans)$.
		\item There exists an $n$-state $\vobs$-correct WFA $\aut$.
	\end{enumerate}
\end{corollary}
\begin{figure}[t]
	\begin{align}
	\eqsys(n,\vobs) \quad{}\triangleq{}\quad
	&
	\underbrace{\bigwedge_{i=1}^{n} \; \bigwedge_{\syma\worda\in\suffixes[\vobs]} \gena_i^{\syma\worda}
		= (\trans^\syma_{i, 1}  \semimul   \gena_1^{\worda}) \semiadd \ldots \semiadd (\trans^\syma_{i, n}  \semimul   \gena_n^{\worda})}_{\text{$\trans$ stabilizes $\gens$ for $\vobs$ (\Cref{def:wit}.\ref{def:wit:stab})}} \notag
	\\
	&{}\wedge{} \underbrace{\bigwedge_{\worda \in \vobs} \mem[\worda] = (\init_1 \semimul \gena_1^{\worda}) \semiadd \ldots \semiadd  (\init_n \semimul \gena_n^{\worda})}_{\text{$\init$ spans $\lang$ for $\vobs$ (\Cref{def:wit}.\ref{def:wit:span})}} \notag
\end{align}
\caption{Equation system $\eqsys(\vobs,n)$ characterizing the existence of $n$-state $\vobs$-correct WFAs. $(\gena_i^{\worda})_{i \in n,\worda\in\suffixes[\vobs]}$, $(\trans^\syma_{i, j})_{\syma\in\syms,(i,j)\in n\times n}$, and $(\init_i)_{i \in n}$ are $\semi$-valued variables.}
\label{fig:eqsys}
\end{figure}
We now arrive at our sought-after system of equations $\eqsys(n, \vobs)$ depicted in \Cref{fig:eqsys}. For all $n\geq 1$ and finite observation sets $\vobs$, $\eqsys(n, \vobs)$  has a solution iff there exists an $n$-dimensional $\vobs$-witness. Moreover, every model of $\eqsys(n, \vobs)$ gives rise to such a witness $\wit$ and therefore also to an $n$-state $\vobs$-correct automaton $\aut_\wit$ by \Cref{thm:wfa_wit_corr}. The crucial observation is that, since $\vobs$ is finite, so is $\suffixes[\vobs]$ and thus all potential generators $\gena_i$ range over a finite domain. The existence of suitable generators $\gena_i$, transition matrices $\trans^\syma$, and initial weights $\iota$ can thus be expressed in terms of a system of equations, where we introduce a variable $\gena_i^{\worda}$ for each application of some $\worda \in \suffixes[\vobs]$ to some generator $\gena_i$.

\begin{example}
\label{ex:enc-swal}
For the $\semitrop$-WFA of \Cref{ex:aut1}, and $\vobs$ = \{\syma\symb\} we have that:
\begin{align*}
\eqsys(2, \vobs) ~{}&\triangleq{}~ 21 = (\init_1 \semimul \gena_1^{\syma\symb}) \semiadd (\init_2 \semimul \gena_2^{\syma\symb}) \\
~{}&\wedge{}~ \gena_1^{\syma\symb} = (\trans^\syma_{1,1} \semimul \gena^\symb_1) \semiadd (\trans^\syma_{1,2} \semimul \gena^\symb_2)
~{}\wedge{}~ \gena_1^{\symb} = (\trans^\symb_{1,1} \semimul \gena^\eword_1) \semiadd (\trans^\symb_{1,2} \semimul \gena^\eword_2) \\
~{}&\wedge{}~ \gena_2^{\syma\symb} = (\trans^\syma_{2,1} \semimul \gena^\symb_1) \semiadd (\trans^\syma_{2,2} \semimul \gena^\symb_2)
~{}\wedge{}~ \gena_2^{\symb} = (\trans^\symb_{2,1} \semimul \gena^\eword_1) \semiadd (\trans^\symb_{2,2} \semimul \gena^\eword_2)
\end{align*}
\end{example}

As \Cref{ex:enc-naive,ex:enc-swal} illustrate, the optimized encoding introduces more equations and more variables than the naive baseline. Counterintuitively,
this apparent overhead is precisely what enables its scalability. As the observation set $\vobs$ grows, the intermediate variables
expose structure shared between observations with matching suffixes, which
the SMT solver can exploit.

\begin{example}
Suppose $\vobs$ grows from $\{\syma\symb\}$ to $\{\syma\symb, \syma\syma\symb\}$.
The naive encoding adds a single, but more complex, equation:
\[
\mem[\syma\syma\symb] ~{}={}~
\begin{pmatrix} \init_1 & \init_2 \end{pmatrix}
\times
\begin{pmatrix}
\trans^\syma_{1,1} & \trans^\syma_{1,2} \\
\trans^\syma_{2,1} & \trans^\syma_{2,2}
\end{pmatrix}
\times
\begin{pmatrix}
\trans^\syma_{1,1} & \trans^\syma_{1,2} \\
\trans^\syma_{2,1} & \trans^\syma_{2,2}
\end{pmatrix}
\times
\begin{pmatrix}
\trans^\symb_{1,1} & \trans^\symb_{1,2} \\
\trans^\symb_{2,1} & \trans^\symb_{2,2}
\end{pmatrix}
\times
\begin{pmatrix} \fin_1 \\ \fin_2 \end{pmatrix}
\]
The optimized encoding, on the other hand, adds one spanning equation
\[
\mem[\syma\syma\symb] ~{}={}~ (\init_1 \semimul \gena_1^{\syma\syma\symb}) \semiadd
(\init_2 \semimul \gena_2^{\syma\syma\symb})
\]
together with two stabilization equations for the new variables
$\gena_1^{\syma\syma\symb},\gena_2^{\syma\syma\symb}$:
\[
\gena_1^{\syma\syma\symb} = (\trans^\syma_{1,1} \semimul \gena_1^{\syma\symb})
\semiadd (\trans^\syma_{1,2} \semimul \gena_2^{\syma\symb}) ~{}\wedge{}~
\gena_2^{\syma\syma\symb} = (\trans^\syma_{2,1} \semimul \gena_1^{\syma\symb})
\semiadd (\trans^\syma_{2,2} \semimul \gena_2^{\syma\symb})
\]
Crucially, the variables $\gena_1^{\syma\symb},\gena_2^{\syma\symb}$ introduced when $\syma\symb$ was
first added to $\vobs$ are reused, so the shared suffix $\syma\symb$ of both
observations is reflected directly in shared subterms of the equation system.
\end{example}

\section{Deciding Equation Systems via SMT Solving}
\label{sec:decidable_semirings}
\noindent
Next, we discuss how to solve the equation systems, as defined in the beginning of \Cref{sec:getaut}, for the semirings considered in this paper. We do this by reducing the problem of computing $\sat{\eqsys}$ to satisfiability problems over decidable theories, which enables using SMT solvers. Every system $\eqsys$ of equations over the Boolean semiring directly corresponds to a propositional satisfiability problem. The encoding is thus straightforward. Let us consider the remaining semirings:

\subsection{The Tropical Semiring}
Solving equation systems $\eqsys$ over the tropical semiring
\[ \semitrop ~{}\triangleq{}~ \langle \mathbb{N} \cup \{\infty\}, \min, +, \infty, 0\rangle\]
 boils down to a satisfiability problem over the decidable theory \theorylia of quantifier-free linear integer arithmetic. Variables are encoded from $\seminat \cup \{ \infty \}$ into $\{-1, 0, 1, 2, \dots \} \subseteq \mathbb{Z}$, where $-1$ is used as a special value representing $\infty$.
The semiring operations are then encoded as:
\begin{align*}
t \semiadd  t' &~{}\triangleq{}~ \mathbf{if}\; \underbrace{t = {-1} \vee t' = {-1}}_{\text{$t$ or $t'$ equals $\infty$}}
\;\mathbf{then}\; \underbrace{\mathit{max}(t, t')}_{{}\text{$t$ if $t'=\infty$ else $t'$}}
\;\mathbf{else}\; \underbrace{\mathit{min}(t, t')}_{\text{usual minimum in $\mathbb{N}$}} \\
t  \semimul  t' &~{}\triangleq{}~ \mathbf{if}\; \underbrace{t = {-1} \vee t' = {-1}}_{\text{$t$ or $t'$ equals $\infty$}}
\;\mathbf{then}\; \underbrace{{-1}}_{t + t'= \infty}
	\;\mathbf{else}\; \underbrace{(t + t')}_{\text{usual sum in $\mathbb{N}$}}~,
\end{align*}
where both $\mathit{max}$ and $\mathit{min}$ are encoded via if-then-else in the obvious way.

\subsection{The Bounded Tropical Semirings}
Let $b\in\mathbb{N}$ and consider the $b$-bounded tropical semiring
\[ \semibtrop{b} ~{}\triangleq{}~ \langle \{0,\ldots,b\} \cup \{\infty\}, \min, \badd{b}, \infty, 0\rangle~.\]
We consider two encodings for solving systems $\eqsys$ of equations over that semiring: A \theorylia encoding similar to the previous section and an encoding \theorybv into the quantifier-free theory of fixed-size bit vectors. Both encode the domain $\{0,\ldots,b\} \cup \{\infty\}$ into the finite set $\{0,\ldots,b,b+1\}$, where $b+1$ represents $\infty$.

\paragraph*{The \theorylia Encoding.}
Since $b+1$ represents $\infty$, we straightforwardly let
\begin{align*}
	t \semiadd t' ~{}\triangleq{}~ \mathit{min}(t, t') \qquad \quad
	t \semimul  t' ~{}\triangleq{}~ \mathit{min}((t + t'), b+1)~.
\end{align*}

\paragraph*{The \theorybv Encoding.}
We map the integer domain $\{0, \dots, b, b + 1\}$ into its (unsigned) binary number representation in \emph{fixed-size bit vectors} of appropriate size $l$, and use fixed-sized bit vector addition and if-then-else to define
\begin{align*}
	t \semiadd t' ~{}\triangleq{}~ \mathit{min}(t, t') \qquad \quad
	t \semimul  t' ~{}\triangleq{}~ \mathit{min}((t + t'), b+1)
\end{align*}
analogously to the \theorylia encoding. We must, however, take care: While just $\lceil \mathit{log}_2(b + 2) \rceil$ bits are sufficient to represent $\{0, \dots, b, b + 1\}$, \emph{intermediate} results from the non-truncated addition $t+t'$ might overflow. The biggest value we must account for is thus $(b+1) + (b+1) = 2(b+1)$, therefore needing to represent $2(b+1)+1 = 2b+3$ values, and so $l = \lceil \mathit{log}_2 ( 2b+3) \rceil$ bits are sufficient to represent all values while avoiding overflows.

\subsection{The Bottleneck Semiring}
Consider the \emph{bottleneck} semiring
\(
\semibot ~{}\triangleq{}~ \langle \seminat \cup \{-\infty, \infty\}, \max, \min, -\infty, \infty\rangle.
\)

We encode the carrier set of $\seminat \cup \{-\infty, \infty\}$ into $\{-2, -1, 0, 1, 2, \dots\} \subseteq \mathbb{Z}$, using $-2$ to represent $-\infty$ and $-1$ for $\infty$. The operations are then encoded in a similar fashion as the tropical operations, accounting for the special infinity values, and minimizing/maximizing otherwise:
\begin{alignat*}{3}
t \semiadd  t' ~{}\triangleq{}~ &\mathbf{if}\; \underbrace{t = {-1} \vee t' = {-1}}_{\text{$t$ or $t'$ equals $\infty$}}\;&&\mathbf{then}\; \underbrace{-1}_{\text{$\max(t, t') = \infty$}} && \\
\;&\mathbf{else}\; \mathbf{if}\; \underbrace{t = {-2}}_{\text{$t$ equals $-\infty$}}\;&&\mathbf{then}\; \underbrace{t'}_{\text{$\max(-\infty, t') = t'$}} && \\
\;&\mathbf{else}\; \mathbf{if}\; \underbrace{t' = {-2}}_{\text{$t'$ equals $-\infty$}}\;&&\mathbf{then}\; \underbrace{t}_{\text{$\max(t, -\infty) = t$}}
\;&&\mathbf{else}\; \underbrace{\max(t, t')}_{\text{$\max$ in $\mathbb{N}$}} \\[3ex]
t \semimul  t' ~{}\triangleq{}~ &\mathbf{if}\; \underbrace{t = {-2} \vee t' = {-2}}_{\text{$t$ or $t'$ equals $-\infty$}}\;&&\mathbf{then}\; \underbrace{-2}_{\text{$\min(t, t') = -\infty$}} && \\
\;&\mathbf{else}\; \mathbf{if}\; \underbrace{t = {-1}}_{\text{$t$ equals $\infty$}}\;&&\mathbf{then}\; \underbrace{t'}_{\text{$\min(\infty, t') = t'$}} && \\
\;&\mathbf{else}\; \mathbf{if}\; \underbrace{t' = {-1}}_{\text{$t'$ equals $\infty$}}\;&&\mathbf{then}\; \underbrace{t}_{\text{$\min(t, \infty) = t$}}
\;&&\mathbf{else}\; \underbrace{\min(t, t')}_{\text{$\min$ in $\mathbb{N}$}}
\end{alignat*}

\begin{remark}[Other Semirings]
The semirings presented above admit linear encodings (\theorybv, \theorylia) of our
equation systems. While this tends to provide performance benefits at the SMT level,
semirings such as the Reals, Viterbi, or Probability semirings (see, e.g., \cite[Tab.\ 1]{Batz2022WeightedProgramming}) could equally be
considered due to their decidable, non-linear (\theorynra) theory. By
\Cref{thm:getaut_for_decidable}, $\fgetaut$ remains complete and $\vobs$-correct
automata can be learned. We view this as a foundation for tackling the
long-standing open problem of learning probabilistic automata, and leave to future
work a systematic exploration leveraging, e.g., approximate SMT solving~\cite{gao_DCompleteDecisionProcedures_2012} or
incremental linearization techniques~\cite{cimatti_IncrementalLinearizationSatisfiability_2018}.
\end{remark}

\section{Experimental Evaluation}
\label{sec:eval}
\noindent
In this section we evaluate our algorithm through a wide range of experiments. We start by describing our prototype implementation, and move on to pose research questions on its performance. We develop a comprehensive set of benchmarks and analyze them to determine the impact of internal optimizations on performance, as well as the algorithm's positioning relative to both the naive baseline (\Cref{sec:getaut:naive}), and a state-of-the-art WFA learning algorithm.

\subsection{Implementation}
To enable an extensive experimental evaluation, we have implemented both \Cref{alg:learn} with the naive baseline from \Cref{sec:getaut:naive} (referred to as $\algonaive$)
and with the state-language-based encoding from \Cref{sec:getaut:non_naive} (referred to as \textbf{S}MT-based \textbf{W}eighted \textbf{A}utomata \textbf{L}earning algorithm ($\algonew$)) in OCaml.

Our implementation supports all semirings discussed in this paper (Boolean, Tropical, bounded Tropical, and Bottleneck), and uses \textsf{Z3}~\cite{de_moura_z3_2008} as a back-end solver. Equation systems over these semirings are decided via the encodings discussed in \Cref{sec:decidable_semirings}. In the following, we discuss our tool's capability of employing incremental SMT solving and how we implement teachers.

\paragraph*{Incremental SMT Solving.}
When the observation set $\vobs$ maintained by \Cref{alg:learn} grows (for fixed $n$), the respective equation systems consist (cf.\ \Cref{sec:getaut:naive} and \Cref{fig:eqsys}) of more and more conjuncts for both $\algonaive$ and $\algonew$. This renders our approach suitable for \emph{incremental} SMT solving. Whether incremental SMT solving is enabled or not can be controlled via a flag. We evaluate the effect of incremental solving in \Cref{sec:ablation}.

\paragraph*{Teachers.}
Our benchmarks aim at learning a given \emph{target WFA $\aut$}, which is known to the teacher but not to the learner. We implement its teacher $\teacher_\aut \triangleq \langle \mem, \eq \rangle$ as follows: Output queries $\mem[\worda]$ are implemented straightforwardly by computing $\aut$'s semantics $\sem{\aut}(\worda)$. Since equivalence queries are generally undecidable~\cite{almagor_whats_2022}, we follow the standard procedure in automata learning~\cite{ferreira_prognosis_2021,fiterau-brostean_combining_2016,aarts_formal_2013,aarts_inference_2010} and resort to \emph{partial equivalence checks}: We let $\eq[\hyp] \triangleq \true$ iff $\sem{\aut}(\worda) = \sem{\hyp}(\worda)$ for the first $N$ words in $\words$ w.r.t.\ a fixed lexicographic order on the symbols, and we let $\eq[\hyp] \triangleq \worda$ only if $\worda$ is the first counterexample found under these $N$ words. For our experiments, we let $N \triangleq 5000 \cdot |\aut| \cdot |\Sigma|$. We have found that this provides a vast cover of reachable automaton behavior\footnote{For the automata learned, on average, counterexamples were found only in the first 5.02\% of this test budget, with over 96.7\% of runs utilizing < 50\% of it.}, while remaining computationally tractable. Developing more advanced testing algorithms for WFAs remains an interesting open area of research, with existing techniques such as the Wp-Method~\cite{fujiwara_TestSelectionBased_1991} being largely motivated for the Boolean setting.

\subsection{Research Questions and Benchmarks}
We have devised three research questions ({\bf RQ}s) for evaluating our algorithms.
\begin{itemize}[leftmargin=*,itemsep=5pt,topsep=10pt]
	\item[]\hspace{-0.28cm} {\bf RQ1}: What is the impact of incremental SMT solving and {bit-vector encodings?}
	\item[]\hspace{-0.315cm} {\bf RQ2}: How does $\algonew$ perform compared to the $\algonaive$ algorithm?
	\item[]\hspace{-0.315cm} {\bf RQ3}: How does $\algonew$ perform compared to the $\algopid$ algorithm?
\end{itemize}
{\bf RQ1} forms an \emph{ablation study} aiming at determining the impact of incremental SMT solving and alternative encodings. {\bf RQ2} enables evaluating whether $\algonew$ is indeed superior to the naive baseline $\algonaive$. In {\bf RQ3},  $\algopid$ refers to a state-of-the-art WFA learning algorithm~\cite{heerdt_learning_2020} for finite semirings, or principal ideal domains. For the semirings considered in this paper, $\algopid$ thus produces (not necessarily minimal) WFAs over the Boolean and bounded Tropical semirings.

\paragraph*{Benchmarks.}
Answering the above-mentioned research questions requires a diverse set of benchmarks, covering a range of different target WFAs for all considered semirings.
Classically, one would rely on existing datasets of realistic targets~\cite{neider_BenchmarksAutomataLearning_2019} for automata learning evaluations.
The line of work in weighted automata learning has not yet contributed such a dataset.
Therefore, we have devised over 5000 randomly generated benchmarks.

Each run was allocated a memory limit of 8GB, and a time limit of 2 hours. Below we describe the range of parameters considered for these experiments. For all semirings considered in this paper, we generate target automata with state sizes ranging from 1 to 12, and alphabet sizes ranging from 2 to 50 symbols.

For $\semibool$-WFAs, we use a well-studied family of minimal NFAs~\cite[Proposition~6.2]{goos_residual_2001}, which is parametric in the number of states and the alphabet size. This yields our scalable set of target WFAs over the Boolean semiring.

For $\semitrop$-WFAs (both unbounded and bounded) and $\semibot$-WFAs, we generated the target WFAs randomly.
Generating random \emph{minimal} automata is, however, famously hard~\cite{hutchison_experimental_2005}, with heuristics only known for DFAs and NFAs over the Boolean semiring. Exact algorithms would rely on generating random automata, and then minimizing them, which is known to be undecidable for $\semitrop$-WFAs. Therefore, we adapt a heuristic from~\cite{hutchison_experimental_2005}: For fixed number of states and alphabet size, we generate transitions in such a way that, for each state $\state$ and symbol $\syma$, the \emph{expected transition density} --- the number of outgoing transitions from $\state$ via $\syma$ with non-$\semizero$ weight --- is $1.25$. For the numeric semirings, these non-$\semizero$ weights are sampled uniformly at random from $\{0,\ldots,\min(1000, b)\}$, where $b$ is either the bound for the bounded Tropical semiring or $\infty$ for the infinite semirings.

\subsection{RQ1: Ablation Study}
\label{sec:ablation}
We start with our ablation study and investigate:
\begin{itemize}[leftmargin=*]
	\item \emph{Incrementality}: What is the impact of \emph{incremental} SMT solving?
	\item \emph{Semiring Encoding}: For the bounded Tropical semiring ($b=100$), what is the impact of the different encodings ($\theorybv$ vs $\theorylia$)?
\end{itemize}

\begin{figure}[t]
	\centering
	\includegraphics[width=0.42\textwidth,alt={\input{./figures/alt/legend-ablation.txt}\unskip}]{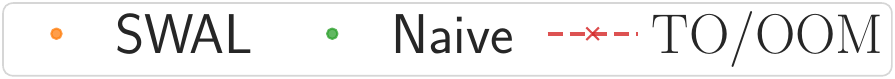}

	\vspace{0.4em}
	\subfloat[Every dot represents a benchmark.]{\label{fig:ablation-incremental}%
		\includegraphics[width=0.49\textwidth,alt={\input{./figures/alt/ablation-incremental.txt}\unskip}]{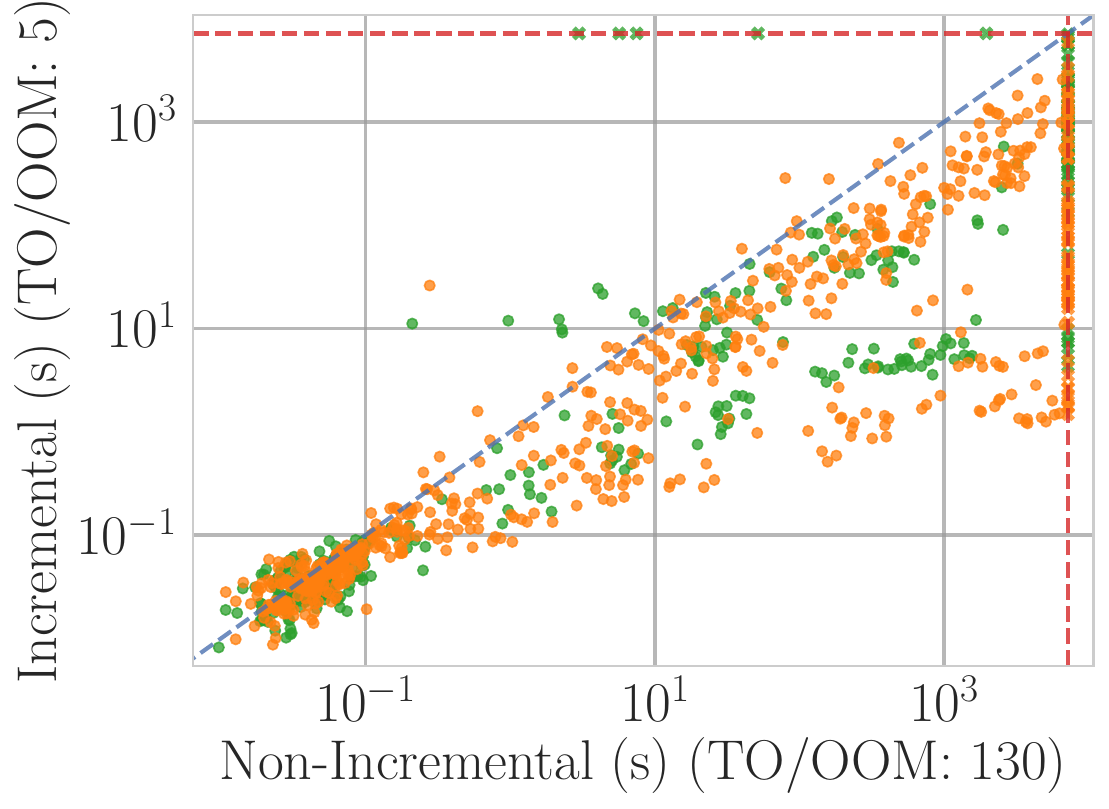}}%
	\hfill
	\subfloat[Every dot represents a benchmark.]{\label{fig:ablation-semiring}%
		\includegraphics[width=0.49\textwidth,alt={\input{./figures/alt/ablation-semiring.txt}\unskip}]{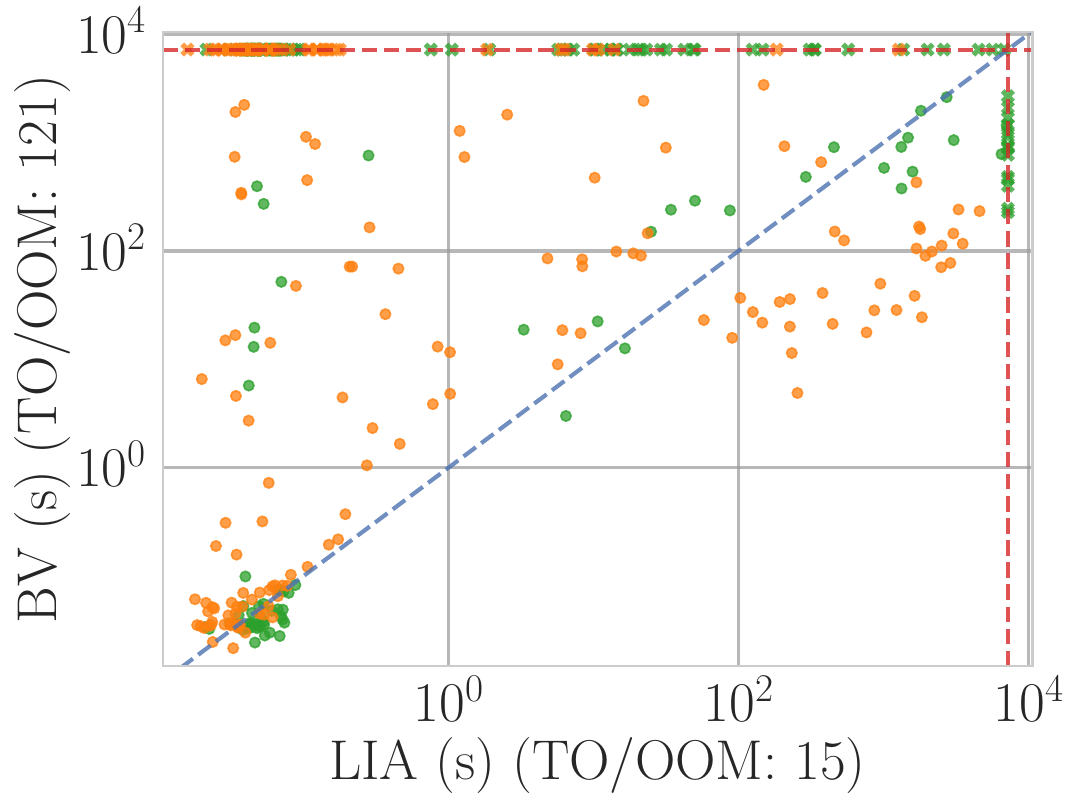}}%

\caption{Effect on learner time (seconds in log-log) of incremental SMT solving (left), and bit vector (\theorybv) vs linear integer arithmetic (\theorylia) encoding for bounded tropical semirings (right). TO/OOM count provided per axis.}
	\label{fig:ablation}
\end{figure}

\paragraph*{Incrementality.}
\Cref{fig:ablation-incremental} indicates the \emph{learner time in seconds}, i.e., the total time spent executing the learner, minus the time spent on computing answers to queries, for both \algonaive and \algonew in log-log scale. We consider \emph{all} benchmarks for all semirings and encodings. Each dot represents one benchmark. The $x$-axis indicates the learner time \emph{without} incremental SMT solving and the $y$-axis \emph{with} incremental solving. Whenever a dot is \emph{below} the diagonal line, incremental solving outperforms non-incremental solving (and vice versa). The red dashed line indicates timeouts / out-of-memory faults (TO/OOM).

We observe that incremental solving almost always results in a shorter learning time, with less benchmarks hitting TO/OOM ($5$ for incremental and $130$ for non-incremental). This effect is observed in both the $\algonaive$ and $\algonew$ algorithms. We conclude that the incremental nature of the way our constraint sets are built is highly beneficial for incremental SMT solving.

\paragraph*{Semiring Encoding.}
\Cref{fig:ablation-semiring} again indicates the learner time in seconds for all benchmarks over the \emph{bounded} ($b=100$) Tropical semiring.  The $x$-axis indicates the learner time when using the \theorylia encoding whereas the $y$-axis indicates the time for the \theorybv encoding. We investigate both \algonaive and \algonew. Whenever a dot is \emph{below} the diagonal line, \theorybv outperforms \theorylia (and vice versa). The red dashed line indicates TO/OOM.

Even though \theorybv occasionally takes a lead over \theorylia, the results indicate that \theorylia mostly outperforms \theorybv. This is particularly striking when considering the number of TO/OOM hits: For \theorybv, $121$ benchmarks hit TO/OOM, whereas for \theorylia, only $15$ benchmarks hit these thresholds.

\subsection{RQ2: Comparison with $\algonaive$ Algorithm}

\begin{figure}[t]
	\centering
	\includegraphics[width=\textwidth,alt={\input{./figures/alt/legend-comparison-naive.txt}\unskip}]{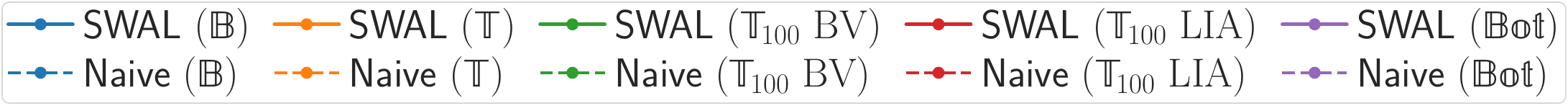}
	\vspace{0.4em}
	\subfloat[]{\label{fig:comparison-naive-time}%
		\includegraphics[width=0.49\textwidth,alt={\input{./figures/alt/comparison-naive-time.txt}\unskip}]{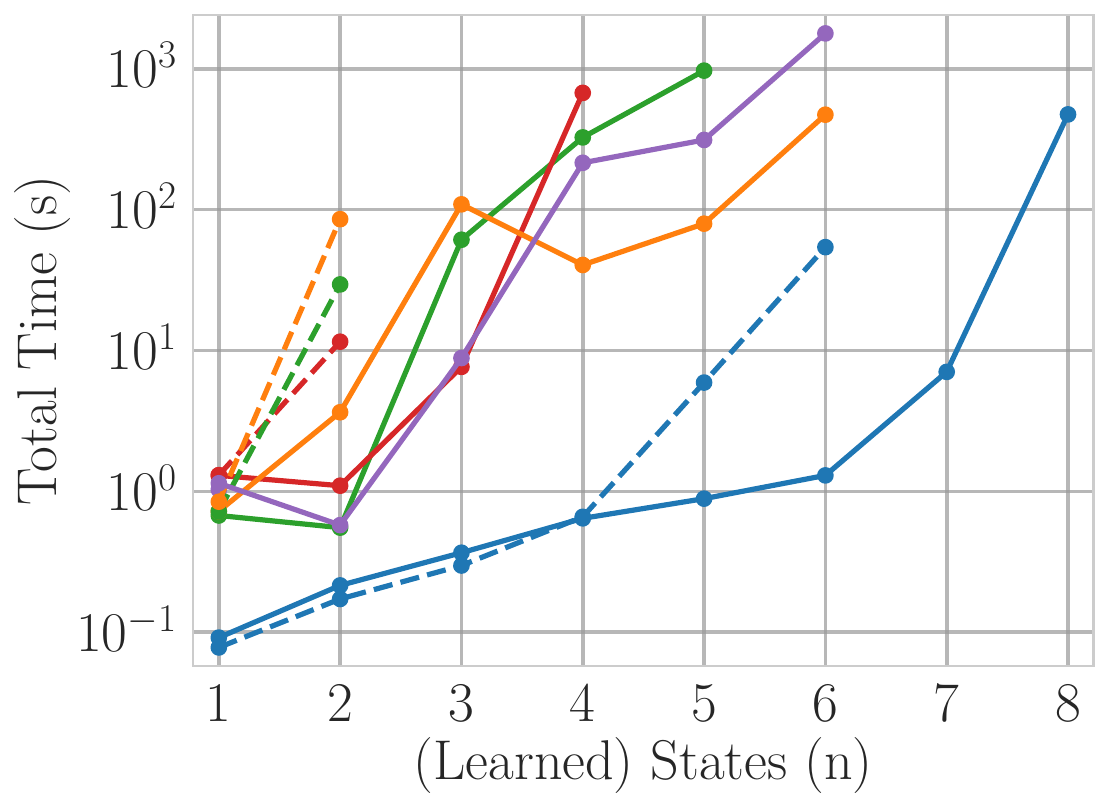}}%
	\hfill
	\subfloat[]{\label{fig:comparison-naive-queries}%
		\includegraphics[width=0.49\textwidth,alt={\input{./figures/alt/comparison-naive-queries.txt}\unskip}]{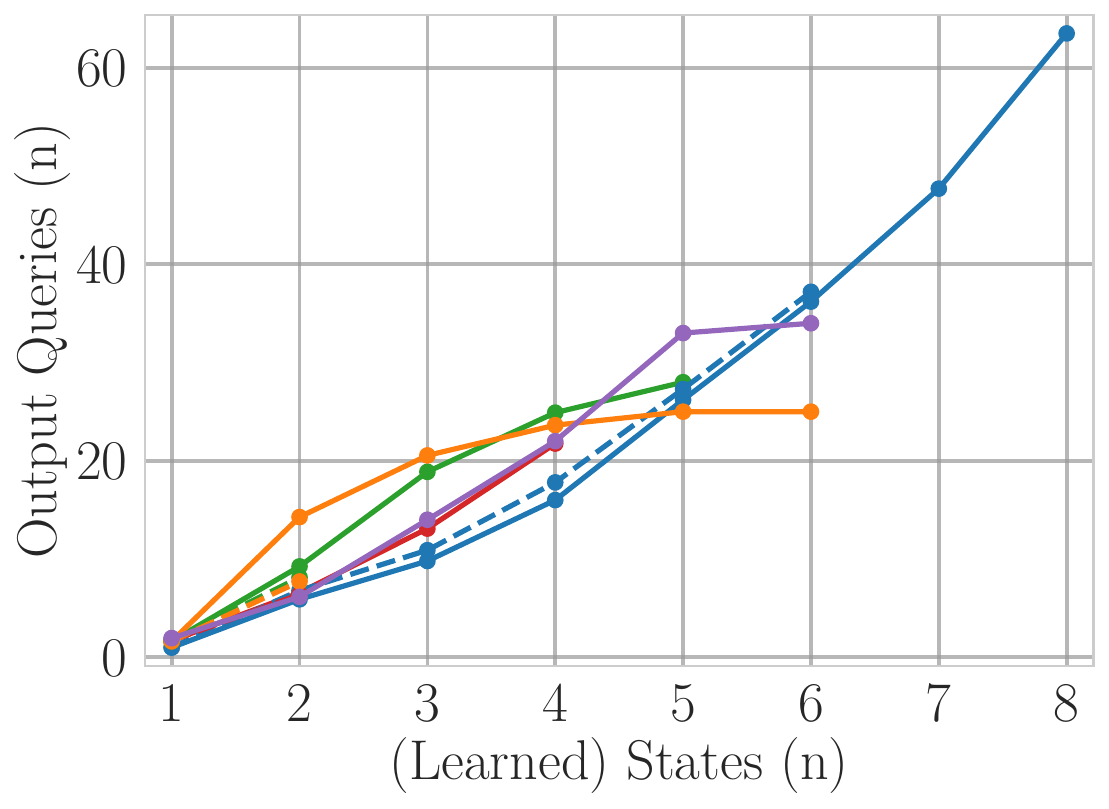}}%

	\caption{Comparison of \algonaive vs \algonew on average total run-time (seconds in log-log, left), and average number of output queries (count, right), as the number of learned states increases. TO/OOM represented by end-of-line.}
	\label{fig:comparison-naive}
\end{figure}

For our comparison of $\algonew$ and $\algonaive$, we compare (i) how the algorithms scale in terms of state sizes of \emph{learned} automata and (ii) their relative query efficiency, i.e., the number of output queries required. We consider benchmarks with incrementality, and $|\syms|=2$ for all semirings. Results are depicted in \Cref{fig:comparison-naive}.

\Cref{fig:comparison-naive-time} shows the \emph{average total time} in seconds (log-scale) on the $y$-axis to learn automata with $n$-states, where $n$ is indicated on the $x$-axis. Dashed lines indicate run-times for \algonaive and solid lines run-times for \algonew, each being distinguished by the respective semiring. We observe that \algonew is capable of learning minimal NFAs (i.e., WFAs over the Boolean semiring) with up to $8$ states within the given time- and memory limits (blue solid line). \algonaive, on the other hand, only learns minimal NFAs with up to $6$ states (blue dashed line). For all other semirings, \algonaive even only scales up to at most $2$ states, whereas \algonew scales up to $6$ states. Moreover, \algonew consistently outperforms \algonaive, sometimes by orders of magnitude. We conclude that \algonew is superior to \algonaive in terms of scalability, and that our more sophisticated equation system based on state-languages is highly beneficial for SMT solving.

\Cref{fig:comparison-naive-queries} shows the \emph{average number of output queries} on the $y$-axis to learn automata with $n$-states, where $n$ is indicated on the $x$-axis. As is to be expected, both \algonaive and \algonew require almost the same number of output queries. Hence, we conclude that  \algonew consistently outperforms \algonaive while not requiring more interaction with the teacher.

\subsection{RQ3: Comparison with $\algopid$ Algorithm}

\begin{figure}[t]
  \centering
	\includegraphics[width=0.57\textwidth,alt={\input{./figures/alt/legend-comparison-angluin.txt}\unskip}]{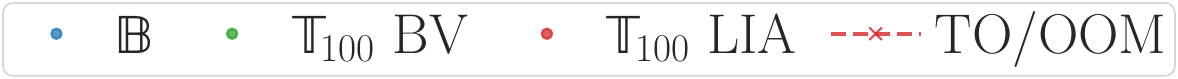}

	\vspace{0.4em}
  \subfloat[]{\label{fig:comparison-angluin-time}%
    \includegraphics[width=0.49\textwidth,alt={\input{./figures/alt/comparison-angluin-time.txt}\unskip}]{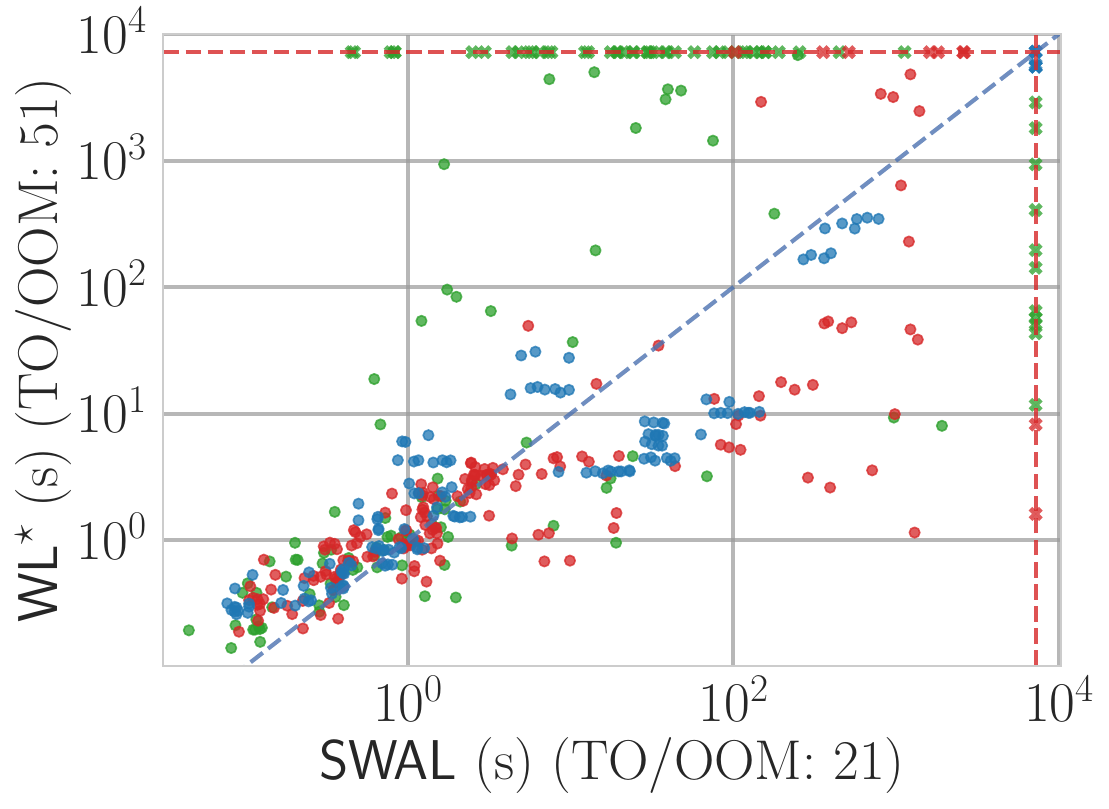}}%
  \hfill
  \subfloat[]{\label{fig:comparison-angluin-queries}%
    \includegraphics[width=0.49\textwidth,alt={\input{./figures/alt/comparison-angluin-queries.txt}\unskip}]{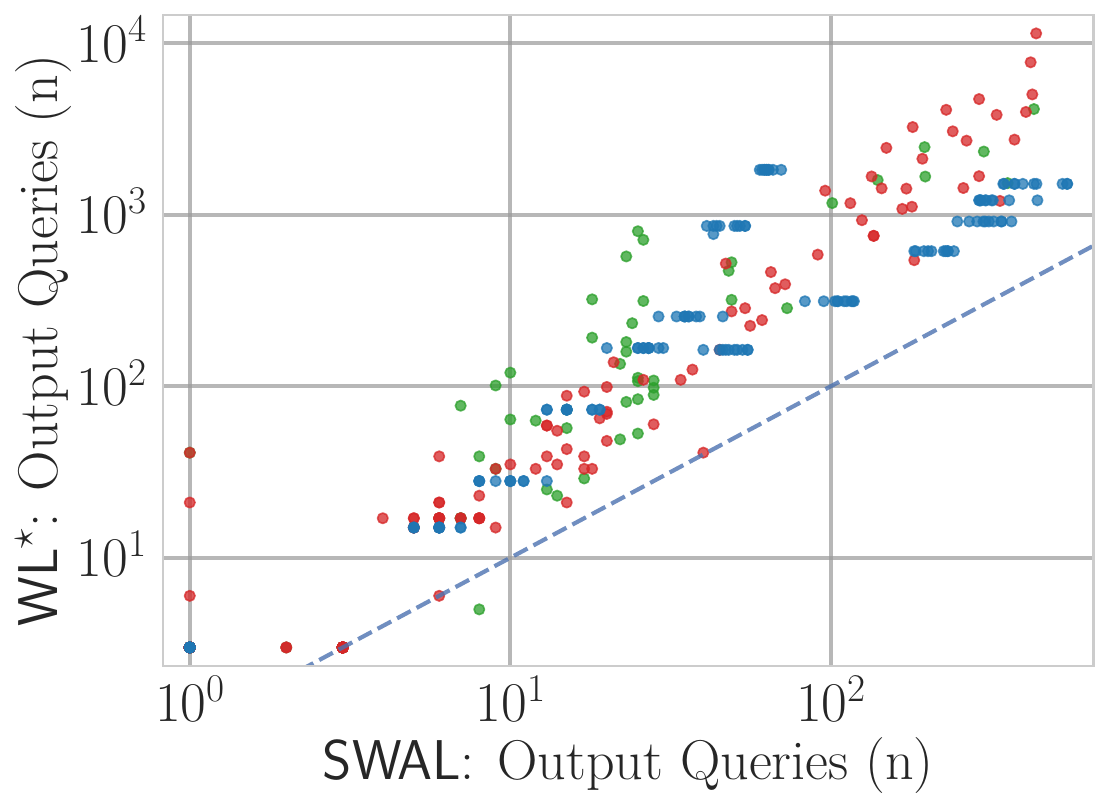}}%

  \subfloat[]{\label{fig:comparison-angluin-size}%
    \includegraphics[width=0.49\textwidth,alt={\input{./figures/alt/comparison-angluin-size.txt}\unskip}]{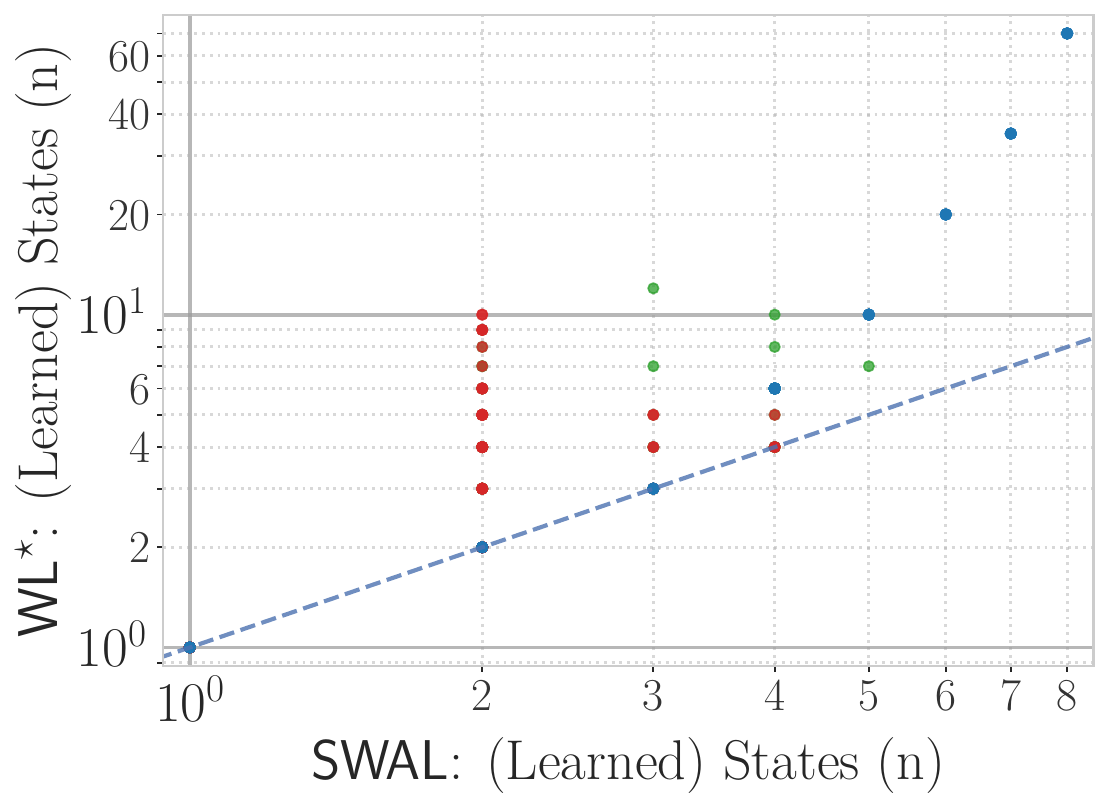}}%
  \vspace{-3pt}
\caption{Comparison of \algopid vs \algonew on total run-time (seconds in log-log, top left), number of output queries (count in log-log, top right), and number of states learned (log-log, bottom). TO/OOM count provided per axis (top left).}
  \label{fig:comparison-angluin}
\end{figure}

$\algopid$~\cite{heerdt_learning_2020} is a state-of-the-art \lstar-based active WFA learning algorithm for semirings forming a principal ideal domain. This means, under the semirings considered in this paper, $\algopid$ can learn WFAs over the Boolean and the \emph{bounded} ($b=100$) Tropical semiring. In contrast to \algonew, $\algopid$ is, however, \emph{not} guaranteed to produce \emph{minimal} automata. We are thus primarily interested in (i) how much the learned automata differ in size and (ii) how many output queries the respective algorithms require for learning these automata.

Since  $\algopid$ has not been implemented, we developed our own reference implementation. To enable a fair comparison, we optimize $\algopid$ in the same manner as we optimized $\algonew$, and use the exact same back-end libraries (for SMT solving and matrix operations), semiring encodings,  and teachers (with random seeds fixed). Importantly, \algopid \emph{always produces automata of the same size, regardless of how it is implemented}, which renders our comparison reasonable.

\paragraph*{Results.}
We consider all benchmarks for the Boolean and the bounded ($b=100$) Tropical semiring. \Cref{fig:comparison-angluin-time} indicates the total run-time in seconds for both \algonew ($x$-axis) and \algopid ($y$-axis) in log-log-scale, distinguished by semiring.  \Cref{fig:comparison-angluin-queries} provides the number of output queries required by the respective algorithms (log-log scale). Whenever a dot is \emph{above} the diagonal line, \algonew needs \emph{less} output queries than \algopid (and vice versa). Finally, \Cref{fig:comparison-angluin-size} gives the state sizes of the \emph{learned} automata for each algorithm. Whenever a dot is \emph{above} the diagonal line, \algopid learned \emph{bigger} automata than \algonew.

We observe that $\algonew$ achieves state-minimality \emph{while consistently requiring fewer output queries} than $\algopid$. The number of states in the learned automata can differ drastically, with automata learned by \algopid being up to $10$ times bigger than the automata learned by \algonew. In terms of total time, more often than not (\Cref{fig:comparison-angluin-time}: 234 vs 163 of 397 terminating points), $\algonew$ terminates faster than $\algopid$, and has a significantly smaller number of TO/OOM hits. We conclude that \algonew consistently learns (sometimes significantly) smaller automata than \algopid while requiring less output queries and remaining competitive on run-time.

\section{Related Work}\label{sec:related}
Active automata learning traces its roots back to the late 1980s with Angluin's seminal \algolstar algorithm~\cite{angluin_learning_1987}. While \algolstar was tailored for learning deterministic finite automata (DFAs), it is the foundation of a long line of adaptations for learning other types of automata, including weighted automata. Bergadano et al.~\cite{bergadano_learning_1996} presented an \algolstar-based algorithm for learning (minimal) weighted automata over fields such as the rationals. Bollig et al. designed the \algonlstar algorithm~\cite{bollig_angluin-style_2009} to learn nondeterministic finite automata (NFAs), which correspond to weighted automata over the Boolean semiring, however the class of learnable automata was restricted (so-called residual automata) and not necessarily minimal. This procedure was generalized by van Heerdt et al. with the \algopid algorithm~\cite{heerdt_learning_2020}, which was proved to terminate for all principal ideal domains (yet no minimality guarantee). Weighted automata learning has generated interest in recent years, with new proofs of termination for the \algopid algorithm for integer-valued automata over the rationals~\cite{buna-marginean_learning_2024}, automata over number fields with almost minimality~\cite{aristote_learning_2025}, and the subclass of deterministic weighted automata over semifields~\cite{pasti_l_2024}. Daviaud et al.~\cite{daviaud_feasability_2025} also provided a more formal language overview of the general feasibility of \algolstar-style learning of weighted automata, however learnability of semirings such as the naturals or tropical semirings remains unknown. In the aforementioned works, only the algorithms from~\cite{bergadano_learning_1996,buna-marginean_learning_2024,aristote_learning_2025} have minimality arguments. The focus is often on using the algebraic properties of the semiring to characterize termination (ignoring minimality) of \algopid by restricting the semiring or the class of weighted languages the WFA can recognize. Our algorithm instead always guarantees minimality, though this comes at the cost of not having an easy way to explore the same algebraic properties for provable guarantees on termination.

Though we are the first to use SMT for active learning of weighted automata, there is previous work on using satisfiability solvers for inference of deterministic finite automata from finite positive and negative examples (in a passive learning setting)~\cite{oliveira,feldman,verwer} and in active learning in the context of an incomplete teacher~\cite{moeller}. Oliveira and Silva~\cite{oliveira}, extending the approach in~\cite{feldman}, use constraint solvers to find a mapping from the states of a prefix-tree automaton (PTA, essentially a partial DFA) to an automaton of a fixed size. If the search fails, the target size is
increased until an automaton is found. The work in~\cite{verwer} takes inspiration from graph-coloring to encode a SAT problem to directly infer DFAs from examples. In the active setting, Tappler and collaborators use SMT encodings to learn timed automata~\cite{tappler}, and Moeller and others~\cite{moeller} formulate a SAT encoding to infer missing entries in an observation table for \lstar -- the motivation being that the teacher might be imperfect or incomplete and not provide answers to all membership queries. All these works are restricted to deterministic automata over the Booleans, whereas our work provides an algorithm in a richer setting of weighted automata over semirings requiring more general solvers.

\section{Conclusion and Future Work}
\label{sec:conclusion}
We have presented a novel SMT-based active automata learning algorithm for WFAs as a practical alternative to Hankel/\lstar-style methods. The evaluation of our implementation shows promise: It learns numerous minimal WFAs over both finite and infinite semirings, many of which are out of scope for existing \lstar-style algorithms, and is competitive with a state-of-the-art algorithm while producing significantly smaller automata and requiring less interaction with the teacher.

As future work, we plan to investigate a more incremental version of our algorithm, where new hypothesis automata are more strongly related to previous ones. Similarly to \lstar-based characterizations of termination, this might yield (i) performance increases and (ii) generalizations of \Cref{thm:flearn_terminate} for termination guarantees for a broader class of semirings. Moreover, we plan to extend our constraint-based approach for learning a broader class of automata including, e.g., probabilistic automata. The existence of active learning for probabilistic languages is a longstanding open question that is closely related to convex optimization problems and it would be interesting to investigate if an SMT approach would offer interesting trade-offs when compared to other approaches. Finally, we plan to apply our technique to the quantitative verification of networks by learning models expressed in Weighted NetKAT \cite{Acevedo2026WeightedNetKAT}.

\begin{credits}
\subsubsection{\ackname}
We thank the CAV reviewers for their thoughtful feedback that helped us
improve the paper. This material is based on work supported by the Defense
Advanced Research Projects Agency (DARPA) under Contract No. HR001125CE018
(approved for public release; distribution unlimited.), the Advanced Research Projects Agency for Health
(ARPA-H), and ERC grant Autoprobe (no. 101002697).
The views and conclusions contained in this document are those of the authors and should
not be interpreted as representing the official policies, either expressed or implied, of the U.S.
Government.
Kevin Batz conducted part of this work at University College London.
\subsubsection{\discintname}
The authors have no competing interests to declare that are relevant to the content of this article.
\subsubsection{Data-Availability Statement}
The full implementation of $\algonew$, as well as resulting experimental data, and analysis scripts used in this paper are available in~\cite{zenodo}.
\end{credits}

\bibliographystyle{splncs04}
\bibliography{refs}

@incollection{heerdt_learning_2020,
	location = {Cham},
	title = {Learning Weighted Automata over Principal Ideal Domains},
	volume = {12077},
	isbn = {978-3-030-45230-8 978-3-030-45231-5},
	url = {http://link.springer.com/10.1007/978-3-030-45231-5_31},
	doi = {10.1007/978-3-030-45231-5_31},
	pages = {602--621},
	booktitle = {Foundations of Software Science and Computation Structures},
	publisher = {Springer International Publishing},
	author = {van Heerdt, Gerco and Kupke, Clemens and Rot, Jurriaan and Silva, Alexandra},
	editor = {Goubault-Larrecq, Jean and König, Barbara},
	urldate = {2025-05-28},
	date = {2020},
	year = {2020},
	langid = {english},
	note = {Series Title: Lecture Notes in Computer Science},
}

@article{angluin_learning_1987,
	title = {Learning regular sets from queries and counterexamples},
	volume = {75},
	issn = {08905401},
	url = {https://linkinghub.elsevier.com/retrieve/pii/0890540187900526},
	doi = {10.1016/0890-5401(87)90052-6},
	pages = {87--106},
	number = {2},
	journal = {Information and Computation},
	journaltitle = {Information and Computation},
	shortjournal = {Information and Computation},
	author = {Angluin, Dana},
	urldate = {2020-12-23},
	date = {1987-11},
	year = {1987},
	langid = {english},
}

@inproceedings{moeller,
  author       = {Mark Moeller and
                  Thomas Wiener and
                  Alaia Solko{-}Breslin and
                  Caleb Koch and
                  Nate Foster and
                  Alexandra Silva},
  editor       = {Karim Ali and
                  Guido Salvaneschi},
  title        = {Automata Learning with an Incomplete Teacher},
  booktitle    = {37th European Conference on Object-Oriented Programming, {ECOOP} 2023,
                  Seattle, Washington, United States, July 17-21, 2023},
  series       = {LIPIcs},
  volume       = {263},
  pages        = {21:1--21:30},
  publisher    = {Schloss Dagstuhl - Leibniz-Zentrum f{\"{u}}r Informatik},
  year         = {2023},
  url          = {https://doi.org/10.4230/LIPIcs.ECOOP.2023.21},
  doi          = {10.4230/LIPICS.ECOOP.2023.21},
  bibsource    = {dblp computer science bibliography, https://dblp.org}
}

@inproceedings{verwer,
  author       = {Marijn Heule and
                  Sicco Verwer},
  editor       = {Jos{\'{e}} M. Sempere and
                  Pedro Garc{\'{\i}}a},
  title        = {Exact {DFA} Identification Using {SAT} Solvers},
  booktitle    = {Grammatical Inference: Theoretical Results and Applications, 10th
                  International Colloquium, {ICGI} 2010, Valencia, Spain, September
                  13-16, 2010. Proceedings},
  series       = {Lecture Notes in Computer Science},
  volume       = {6339},
  pages        = {66--79},
  publisher    = {Springer},
  year         = {2010},
  url          = {https://doi.org/10.1007/978-3-642-15488-1\_7},
  doi          = {10.1007/978-3-642-15488-1\_7},
  bibsource    = {dblp computer science bibliography, https://dblp.org}
}

@ARTICLE{feldman,
  author={Biermann, A. W. and Feldman, J. A.},
  journal={IEEE Transactions on Computers},
  title={On the Synthesis of Finite-State Machines from Samples of Their Behavior},
  year={1972},
  volume={C-21},
  number={6},
  pages={592-597},
  doi={10.1109/TC.1972.5009015}}

@article{oliveira,
  author       = {Arlindo L. Oliveira and
                  Jo{\~{a}}o P. Marques Silva},
  title        = {Efficient Algorithms for the Inference of Minimum Size DFAs},
  journal      = {Mach. Learn.},
  volume       = {44},
  number       = {1/2},
  pages        = {93--119},
  year         = {2001},
  url          = {https://doi.org/10.1023/A:1010828029885},
  doi          = {10.1023/A:1010828029885},
  bibsource    = {dblp computer science bibliography, https://dblp.org}
}

@article{vaandrager_model_2017,
	title = {Model Learning},
	volume = {60},
	doi = {10.1145/2967606},
	pages = {86--95},
	journal = {Communications of the {ACM}},
	journaltitle = {Communications of the {ACM}},
	shortjournal = {Commun. {ACM}},
	author = {Vaandrager, Frits W.},
	date = {2017},
	year = {2017},
}

@inproceedings{fiterau-brostean_combining_2016,
	title = {Combining Model Learning and Model Checking to Analyze {TCP} Implementations},
	volume = {9780},
	doi = {10.1007/978-3-319-41540-6_25},
	series = {{LNCS}},
	pages = {454--471},
	booktitle = {{CAV}},
	publisher = {Springer},
	author = {Fiterau-Brostean, Paul and Janssen, Ramon and Vaandrager, Frits W.},
	date = {2016},
	year = {2016},
}

@inproceedings{ferreira_prognosis_2021,
	title = {Prognosis: closed-box analysis of network protocol implementations},
	doi = {10.1145/3452296.3472938},
	pages = {762--774},
	booktitle = {{SIGCOMM}},
	publisher = {{ACM}},
	author = {Ferreira, Tiago and Brewton, Harrison and D'Antoni, Loris and Silva, Alexandra},
	date = {2021},
}

@incollection{finkbeiner_learning_2015,
	location = {Cham},
	title = {Learning the Language of Error},
	volume = {9364},
	isbn = {978-3-319-24952-0 978-3-319-24953-7},
	url = {http://link.springer.com/10.1007/978-3-319-24953-7_9},
	doi = {10.1007/978-3-319-24953-7_9},
	pages = {114--130},
	booktitle = {Automated Technology for Verification and Analysis},
	publisher = {Springer International Publishing},
	author = {Chapman, Martin and Chockler, Hana and Kesseli, Pascal and Kroening, Daniel and Strichman, Ofer and Tautschnig, Michael},
	editor = {Finkbeiner, Bernd and Pu, Geguang and Zhang, Lijun},
	urldate = {2026-01-29},
	date = {2015},
	year = {2015},
	note = {Series Title: Lecture Notes in Computer Science},
}

@inproceedings{aarts_inference_2010,
	location = {Berlin, Heidelberg},
	title = {Inference and Abstraction of the Biometric Passport},
	isbn = {978-3-642-16558-0},
	doi = {10.1007/978-3-642-16558-0_54},
	series = {Lecture Notes in Computer Science},
	pages = {673--686},
	booktitle = {Leveraging Applications of Formal Methods, Verification, and Validation},
	publisher = {Springer},
	author = {Aarts, Fides and Schmaltz, Julien and Vaandrager, Frits},
	editor = {Margaria, Tiziana and Steffen, Bernhard},
	date = {2010},
	year = {2010},
	langid = {english},
}

@inproceedings{marksteiner_automated_2024,
	location = {New York, {NY}, {USA}},
	title = {Automated Passport Control: Mining and Checking Models of Machine Readable Travel Documents},
	isbn = {979-8-4007-1718-5},
	url = {https://dl.acm.org/doi/10.1145/3664476.3670454},
	doi = {10.1145/3664476.3670454},
	series = {{ARES} '24},
	shorttitle = {Automated Passport Control},
	pages = {1--8},
	booktitle = {Proceedings of the 19th International Conference on Availability, Reliability and Security},
	publisher = {Association for Computing Machinery},
	author = {Marksteiner, Stefan and Sirjani, Marjan and Sjödin, Mikael},
	urldate = {2024-07-31},
	date = {2024-07-30},
	year = {2024},
}

@inproceedings{aarts_formal_2013,
	title = {Formal Models of Bank Cards for Free},
	doi = {10.1109/ICSTW.2013.60},
	eventtitle = {2013 {IEEE} Sixth International Conference on Software Testing, Verification and Validation Workshops},
	pages = {461--468},
	booktitle = {2013 {IEEE} Sixth International Conference on Software Testing, Verification and Validation Workshops},
	author = {Aarts, F. and Ruiter, J. De and Poll, E.},
	date = {2013-03},
	year = {2013},
}

@incollection{wu_black_1999,
	location = {Boston, {MA}},
	title = {Black Box Checking},
	volume = {28},
	isbn = {978-1-4757-5270-0 978-0-387-35578-8},
	url = {http://link.springer.com/10.1007/978-0-387-35578-8_13},
	doi = {10.1007/978-0-387-35578-8_13},
	pages = {225--240},
	booktitle = {Formal Methods for Protocol Engineering and Distributed Systems},
	publisher = {Springer {US}},
	author = {Peled, Doron and Vardi, Moshe Y. and Yannakakis, Mihalis},
	editor = {Wu, Jianping and Chanson, Samuel T. and Gao, Qiang},
	urldate = {2026-01-29},
	date = {1999},
	year = {1999},
	langid = {english},
	note = {Series Title: {IFIP} Advances in Information and Communication Technology},
}

@thesis{heerdt_efficient_2014,
	type = {Bachelor's Thesis},
	title = {Efficient Inference of Mealy Machines},
	institution = {Radboud University Nijmegen},
	author = {van Heerdt, Gerco},
	date = {2014},
	year = {2014},
}

@incollection{finkbeiner_scalable_2024,
	location = {Cham},
	title = {Scalable Tree-based Register Automata Learning},
	volume = {14571},
	isbn = {978-3-031-57248-7 978-3-031-57249-4},
	url = {https://link.springer.com/10.1007/978-3-031-57249-4_5},
	doi = {10.1007/978-3-031-57249-4_5},
	pages = {87--108},
	booktitle = {Tools and Algorithms for the Construction and Analysis of Systems},
	publisher = {Springer Nature Switzerland},
	author = {Dierl, Simon and Fiterau-Brostean, Paul and Howar, Falk and Jonsson, Bengt and Sagonas, Konstantinos and Tåquist, Fredrik},
	editor = {Finkbeiner, Bernd and Kovács, Laura},
	urldate = {2025-04-03},
	date = {2024},
	langid = {english},
	note = {Series Title: Lecture Notes in Computer Science},
}

@article{bergadano_learning_1996,
	title = {Learning Behaviors of Automata from Multiplicity and Equivalence Queries},
	volume = {25},
	issn = {0097-5397, 1095-7111},
	url = {http://epubs.siam.org/doi/10.1137/S009753979326091X},
	doi = {10.1137/S009753979326091X},
	pages = {1268--1280},
	number = {6},
	journal = {{SIAM} Journal on Computing},
	journaltitle = {{SIAM} Journal on Computing},
	shortjournal = {{SIAM} J. Comput.},
	author = {Bergadano, Francesco and Varricchio, Stefano},
	urldate = {2026-01-29},
	date = {1996-12},
	year = {1996},
	langid = {english},
}

@inproceedings{aristote_learning_2025,
	location = {Singapore, Singapore},
	title = {Learning Weighted Automata over Number Rings, Concretely and Categorically},
	isbn = {979-8-3315-7900-5},
	url = {https://ieeexplore.ieee.org/document/11186332/},
	doi = {10.1109/LICS65433.2025.00038},
	eventtitle = {2025 40th Annual {ACM}/{IEEE} Symposium on Logic in Computer Science ({LICS})},
	pages = {417--430},
	booktitle = {2025 40th Annual {ACM}/{IEEE} Symposium on Logic in Computer Science ({LICS})},
	publisher = {{IEEE}},
	author = {Aristote, Quentin and Van Gool, Sam and Petrişan, Daniela and Shirmohammadi, Mahsa},
	urldate = {2026-01-29},
	date = {2025-06-23},
	year = {2025},
}

@collection{droste_handbook_2009,
	location = {Berlin, Heidelberg},
	title = {Handbook of Weighted Automata},
	isbn = {978-3-642-01491-8 978-3-642-01492-5},
	url = {https://link.springer.com/10.1007/978-3-642-01492-5},
	doi = {10.1007/978-3-642-01492-5},
	series = {Monographs in Theoretical Computer Science. An {EATCS} Series},
	publisher = {Springer Berlin Heidelberg},
	editor = {Droste, Manfred and Kuich, Werner and Vogler, Heiko},
	date = {2009},
	year = {2009},
	langid = {english},
}

@article{bonsangue_sound_2013,
	title = {Sound and Complete Axiomatizations of Coalgebraic Language Equivalence},
	volume = {14},
	issn = {1529-3785, 1557-945X},
	url = {https://dl.acm.org/doi/10.1145/2422085.2422092},
	doi = {10.1145/2422085.2422092},
	pages = {1--52},
	number = {1},
	journal = {{ACM} Transactions on Computational Logic},
	journaltitle = {{ACM} Transactions on Computational Logic},
	shortjournal = {{ACM} Trans. Comput. Logic},
	author = {Bonsangue, Marcello M. and Milius, Stefan and Silva, Alexandra},
	urldate = {2026-01-29},
	date = {2013-02},
	year = {2013},
	langid = {english},
}

@inproceedings{de_moura_z3_2008,
	location = {Berlin, Heidelberg},
	title = {Z3: An Efficient {SMT} Solver},
	isbn = {978-3-540-78800-3},
	doi = {10.1007/978-3-540-78800-3_24},
	series = {Lecture Notes in Computer Science},
	shorttitle = {Z3},
	pages = {337--340},
	booktitle = {Tools and Algorithms for the Construction and Analysis of Systems},
	publisher = {Springer},
	author = {de Moura, Leonardo and Bjørner, Nikolaj},
	editor = {Ramakrishnan, C. R. and Rehof, Jakob},
	date = {2008},
	langid = {english},
}

@article{almagor_whats_2022,
	title = {What's decidable about weighted automata?},
	volume = {282},
	issn = {08905401},
	url = {https://linkinghub.elsevier.com/retrieve/pii/S0890540120301395},
	doi = {10.1016/j.ic.2020.104651},
	pages = {104651},
	journal = {Information and Computation},
	journaltitle = {Information and Computation},
	shortjournal = {Information and Computation},
	author = {Almagor, Shaull and Boker, Udi and Kupferman, Orna},
	urldate = {2025-02-03},
	date = {2022-01},
	langid = {english},
}

@incollection{hutchison_experimental_2005,
	location = {Berlin, Heidelberg},
	title = {Experimental Evaluation of Classical Automata Constructions},
	volume = {3835},
	isbn = {978-3-540-30553-8 978-3-540-31650-3},
	url = {http://link.springer.com/10.1007/11591191_28},
	doi = {10.1007/11591191_28},
	pages = {396--411},
	booktitle = {Logic for Programming, Artificial Intelligence, and Reasoning},
	publisher = {Springer Berlin Heidelberg},
	author = {Tabakov, Deian and Vardi, Moshe Y.},
	editor = {Sutcliffe, Geoff and Voronkov, Andrei},
	editorb = {Hutchison, David and Kanade, Takeo and Kittler, Josef and Kleinberg, Jon M. and Mattern, Friedemann and Mitchell, John C. and Naor, Moni and Nierstrasz, Oscar and Pandu Rangan, C. and Steffen, Bernhard and Sudan, Madhu and Terzopoulos, Demetri and Tygar, Dough and Vardi, Moshe Y. and Weikum, Gerhard},
	editorbtype = {redactor},
	urldate = {2026-01-07},
	date = {2005},
	langid = {english},
	note = {Series Title: Lecture Notes in Computer Science},
}

@incollection{goos_residual_2001,
	location = {Berlin, Heidelberg},
	title = {Residual Finite State Automata},
	volume = {2010},
	isbn = {978-3-540-41695-1 978-3-540-44693-4},
	url = {http://link.springer.com/10.1007/3-540-44693-1_13},
	doi = {10.1007/3-540-44693-1_13},
	pages = {144--157},
	booktitle = {{STACS} 2001},
	publisher = {Springer Berlin Heidelberg},
	author = {Denis, François and Lemay, Aurélien and Terlutte, Alain},
	editor = {Ferreira, Afonso and Reichel, Horst},
	editorb = {Goos, Gerhard and Hartmanis, Juris and Van Leeuwen, Jan},
	editorbtype = {redactor},
	urldate = {2025-09-17},
	date = {2001},
	note = {Series Title: Lecture Notes in Computer Science},
}

@inproceedings{bollig_angluin-style_2009,
	location = {San Francisco, {CA}, {USA}},
	title = {Angluin-style learning of {NFA}},
	series = {{IJCAI}'09},
	pages = {1004--1009},
	booktitle = {Proceedings of the 21st International Joint Conference on Artificial Intelligence},
	publisher = {Morgan Kaufmann Publishers Inc.},
	author = {Bollig, Benedikt and Habermehl, Peter and Kern, Carsten and Leucker, Martin},
	date = {2009},
}

@article{buna-marginean_learning_2024,
	title = {On Learning Polynomial Recursive Programs},
	volume = {8},
	issn = {2475-1421},
	url = {https://dl.acm.org/doi/10.1145/3632876},
	doi = {10.1145/3632876},
	pages = {1001--1027},
	issue = {{POPL}},
	journal = {Proceedings of the {ACM} on Programming Languages},
	journaltitle = {Proceedings of the {ACM} on Programming Languages},
	shortjournal = {Proc. {ACM} Program. Lang.},
	author = {Buna-Marginean, Alex and Cheval, Vincent and Shirmohammadi, Mahsa and Worrell, James},
	urldate = {2026-01-29},
	date = {2024-01-02},
	langid = {english},
}

@inproceedings{pasti_l_2024,
	location = {Miami, Florida, {USA}},
	title = {An L* Algorithm for Deterministic Weighted Regular Languages},
	url = {https://aclanthology.org/2024.emnlp-main.468},
	doi = {10.18653/v1/2024.emnlp-main.468},
	eventtitle = {Proceedings of the 2024 Conference on Empirical Methods in Natural Language Processing},
	pages = {8197--8210},
	booktitle = {Proceedings of the 2024 Conference on Empirical Methods in Natural Language Processing},
	publisher = {Association for Computational Linguistics},
	author = {Pasti, Clemente and Karagöz, Talu and Nowak, Franz and Svete, Anej and Boumasmoud, Reda and Cotterell, Ryan},
	urldate = {2026-01-29},
	date = {2024},
	langid = {english},
}

@article{daviaud_feasability_2025,
	title = {Feasability of Learning Weighted Automata on a Semiring},
	volume = {Volume 21, Issue 3},
	issn = {1860-5974},
	url = {https://doi.org/10.46298/lmcs-21(3:15)2025},
	doi = {10.46298/lmcs-21(3:15)2025},
	pages = {13612},
	journal = {Logical Methods in Computer Science},
	journaltitle = {Logical Methods in Computer Science},
	author = {Daviaud, Laure and Johnson, Marianne},
	urldate = {2025-09-17},
	date = {2025-08-07},
	eprinttype = {arxiv},
	eprint = {2309.07806 [cs]},
}

@incollection{tappler,
  title = {Timed {{Automata Learning}} via {{SMT Solving}}},
  booktitle = {{{NASA Formal Methods}}},
  author = {Tappler, Martin and Aichernig, Bernhard K. and Lorber, Florian},
  editor = {Deshmukh, Jyotirmoy V. and Havelund, Klaus and Perez, Ivan},
  date = {2022},
  volume = {13260},
  pages = {489--507},
  publisher = {Springer International Publishing},
  location = {Cham},
  doi = {10.1007/978-3-031-06773-0_26},
  isbn = {978-3-031-06772-3 978-3-031-06773-0},
  langid = {english}
}

@incollection{neider_BenchmarksAutomataLearning_2019,
  title = {Benchmarks for {{Automata Learning}} and {{Conformance Testing}}},
  booktitle = {Models, {{Mindsets}}, {{Meta}}: {{The What}}, the {{How}}, and the {{Why Not}}?},
  author = {Neider, Daniel and Smetsers, Rick and Vaandrager, Frits and Kuppens, Harco},
  editor = {Margaria, Tiziana and Graf, Susanne and Larsen, Kim G.},
  date = {2019},
  volume = {11200},
  pages = {390--416},
  publisher = {Springer International Publishing},
  location = {Cham},
  doi = {10.1007/978-3-030-22348-9_23},
  isbn = {978-3-030-22347-2 978-3-030-22348-9},
  langid = {english}
}

@misc{zenodo,
  title = {{{SMT-Based Active Learning}} of {{Weighted Automata}} (Artifact)},
  author = {Ferreira, Tiago and Batz, Kevin and Silva, Alexandra},
  year = 2026,
  month = apr,
  doi = {10.5281/ZENODO.19701329},
  urldate = {2026-04-23},
  copyright = {MIT License},
  howpublished = {Zenodo},
  langid = {english}
}

@incollection{gao_DCompleteDecisionProcedures_2012,
  title = {$\delta$-{{Complete Decision Procedures}} for {{Satisfiability}} over the {{Reals}}},
  booktitle = {Automated {{Reasoning}}},
  author = {Gao, Sicun and Avigad, Jeremy and Clarke, Edmund M.},
  editor = {Gramlich, Bernhard and Miller, Dale and Sattler, Uli},
  editora = {Hutchison, David and Kanade, Takeo and Kittler, Josef and Kleinberg, Jon M. and Mattern, Friedemann and Mitchell, John C. and Naor, Moni and Nierstrasz, Oscar and Pandu Rangan, C. and Steffen, Bernhard and Sudan, Madhu and Terzopoulos, Demetri and Tygar, Doug and Vardi, Moshe Y. and Weikum, Gerhard},
  editoratype = {redactor},
  date = {2012},
  volume = {7364},
  pages = {286--300},
  publisher = {Springer Berlin Heidelberg},
  location = {Berlin, Heidelberg},
  doi = {10.1007/978-3-642-31365-3_23},
  isbn = {978-3-642-31364-6 978-3-642-31365-3}
}

@article{cimatti_IncrementalLinearizationSatisfiability_2018,
  title = {Incremental {{Linearization}} for {{Satisfiability}} and {{Verification Modulo Nonlinear Arithmetic}} and {{Transcendental Functions}}},
  author = {Cimatti, Alessandro and Griggio, Alberto and Irfan, Ahmed and Roveri, Marco and Sebastiani, Roberto},
  year = 2018,
  month = jul,
  journal = {ACM Transactions on Computational Logic},
  volume = {19},
  number = {3},
  pages = {1--52},
  issn = {1529-3785, 1557-945X},
  doi = {10.1145/3230639},
  urldate = {2026-04-24},
  langid = {english}
}

@article{fujiwara_TestSelectionBased_1991,
  title = {Test Selection Based on Finite State Models},
  author = {Fujiwara, S. and V. Bochmann, G. and Khendek, F. and Amalou, M. and Ghedamsi, A.},
  date = {1991-06},
  journaltitle = {IEEE Transactions on Software Engineering},
  shortjournal = {IIEEE Trans. Software Eng.},
  volume = {17},
  number = {6},
  pages = {591--603},
  issn = {00985589},
  doi = {10.1109/32.87284}
}

@article{Batz2022WeightedProgramming,
	author    = {Kevin Batz and Adrian Gallus and Benjamin Lucien Kaminski and Joost-Pieter Katoen and Tobias Winkler},
	title     = {Weighted Programming: A Programming Paradigm for Specifying Mathematical Models},
	journal   = {Proceedings of the ACM on Programming Languages},
	volume    = {6},
	number    = {OOPSLA1},
	articleno = {66},
	pages     = {1--30},
	year      = {2022},
	doi       = {10.1145/3527310}
}

@article{Acevedo2026WeightedNetKAT,
    author    = {Emmanuel {Su\'arez Acevedo} and Tiago Ferreira and Kevin Batz and Oliver B{\o}ving and Nate Foster and Alexandra Silva},
    title     = {Weighted {NetKAT}: A Programming Language For Quantitative Network Verification},
    journal   = {Proceedings of the {ACM} on Programming Languages},
    journaltitle = {Proceedings of the {ACM} on Programming Languages},
    shortjournal = {Proc. {ACM} Program. Lang.},
    volume    = {10},
    number    = {PLDI},
    issue     = {PLDI},
    articleno = {240},
    article   = {240},
    year      = {2026},
    month     = jun,
    doi       = {10.1145/3808318},
    langid = {english},
    eprint    = {2604.13987},
    eprinttype = {arxiv},
    eprintclass = {cs.PL},
}

\end{document}